\documentclass[11pt]{article}
\usepackage{fullpage}
\usepackage{amsmath,amsthm,amsfonts,amssymb,amscd}
\usepackage{enumitem}
\usepackage{color}
\usepackage
[
colorlinks=true,
urlcolor=blue,
linkcolor=blue,
citecolor=blue,
]
{hyperref}

\newcommand{\rparagraph}[1]{ \paragraph{#1}}

\newtheorem{lemma}{Lemma}
\newtheorem{theorem}[lemma]{Theorem}

\newtheorem{definition}[lemma]{Definition}
\newtheorem{proposition}[lemma]{Proposition}

\newtheorem{claim}[lemma]{Claim}

\numberwithin{lemma}{section}
\numberwithin{question}{section}

\newtheorem{informal theorem}[lemma]{Informal Theorem}
\newtheorem{open problem}[lemma]{Open Problem}

\newcommand{\remph}[1]{\textsf{#1}}

\newcommand{\set}[1]{\left \{{#1} \right \}}
\newcommand{\poly}{\operatorname{poly}}

\newcommand{\remove}[1]{}

\newcommand{\B}{\{0,1\}}
\newcommand{\ar}{\rightarrow}

\newcommand{\eps}{\epsilon}
\newcommand{\from}{\leftarrow}

\DeclareMathOperator*{\Exp}{\mathbb{E}}

\newcommand{\Hi}{H_{\infty}}

\newcommand{\half}{\frac{1}{2}}

\newcommand{\Con}{\mathsf{Con}}
\newcommand{\Red}{\mathsf{Red}}

\newcommand{\PRG}{\mathsf{PRG}}
\newcommand{\HF}{\mathsf{hard}\mbox{-}\mathsf{function}}
\newcommand{\PEG}{\mathsf{PEG}}

\newcommand{\PRGB}{\mathsf{PRG}\mbox{-}\mathsf{breaks}}
\newcommand{\HFB}{\mathsf{hard}\mbox{-}\mathsf{function}\mbox{-}\mathsf{breaks}}
\newcommand{\PEGB}{\mathsf{PEG}\mbox{-}\mathsf{breaks}}

\newcommand{\HFtoPRG}{\HF\Rightarrow\PRG}
\newcommand{\rhoHFtoPRG}[1]{{#1}\mbox{-}\HF\Rightarrow\PRG}
\newcommand{\HFtoepsPRG}[1]{\HF\Rightarrow{#1}\mbox{-}\PRG}
\newcommand{\rhoHFtoepsPRG}[2]{{#1}\mbox{-}\HF\Rightarrow{#2}\mbox{-}\PRG}

\newcommand{\HFtoPEG}{\HF\Rightarrow\PEG}

\newcommand{\rhoHFtokepsPEG}[3]{{#1}\mbox{-}\HF\Rightarrow({#2},{#3})\mbox{-}\PEG}

\newcommand{\HFtorhoHF}[1]{\HF\Rightarrow{#1}\mbox{-}\HF}
\newcommand{\rhoHFtorhoHF}[2]{{#1}\mbox{-}\HF\Rightarrow{#2}\mbox{-}\HF}
\newcommand{\rhoHF}[1]{{#1}\mbox{-}\HF}

\newcommand{\Ti}{\operatorname{Time}}
\newcommand{\Tit}{\Ti_{\operatorname{all}}}
\newcommand{\Fix}{\operatorname{Fix}}

\newcommand{\Func}{\mathcal{F}}

\newcommand{\RndE}{\mathsf{RndE}}

\newcommand{\Free}{\operatorname{Free}}
\newcommand{\Fill}{\operatorname{Fill}}
\newcommand{\Select}{\operatorname{Select}}
\newcommand{\Assign}{\operatorname{Assign}}

\newcommand{\Rest}{\mathsf{R}}
\newcommand{\wt}{\operatorname{weight}}

\newcommand{\Enc}{\operatorname{Enc}}
\newcommand{\Dec}{\operatorname{Dec}}
\newcommand{\shared}{\operatorname{shared}}
\newcommand{\List}{\operatorname{List}}

\newcommand{\RNSY}{\operatorname{RNSY}}

\title{On Hardness Assumptions Needed for ``Extreme High-End'' PRGs and Fast Derandomization}
\author{Ronen Shaltiel\thanks{Department of computer science, University of Haifa. E-mail: \texttt{ronen@cs.haifa.ac.il}. This research was supported by ISF grants 1628/17 and 1006/23.} \and Emanuele Viola\thanks{Khoury College of Computer Sciences, Northeastern University, E-mail: \texttt{viola@ccs.neu.edu}. Supported by NSF CCF award 1813930 and NSF CCF award 2114116.}}

\begin{document}

\begin{titlepage}
\maketitle

\begin{abstract}
The hardness vs.~randomness paradigm aims to explicitly construct pseudorandom generators $G:\{0,1\}^r \to \{0,1\}^m$ that fool circuits of size $m$, assuming the existence of explicit hard functions. A ``high-end PRG'' with seed length $r=O(\log m)$ (implying BPP=P) was achieved in a seminal work of Impagliazzo and Wigderson (STOC 1997), assuming \textsc{the high-end hardness assumption}: there exist constants $0<\beta < 1< B$, and functions computable in time $2^{B \cdot n}$ that cannot be computed by circuits of size $2^{\beta \cdot n}$.

Recently, motivated by fast derandomization of randomized algorithms, Doron et al.~(FOCS 2020) and Chen and Tell (STOC 2021), construct ``extreme high-end PRGs'' with seed length $r=(1+o(1))\cdot \log m$, under qualitatively stronger assumptions.

We study whether extreme high-end PRGs can be constructed from the following scaled version of the assumption which we call \textsc{the extreme high-end hardness assumption}, and
in which $\beta=1-o(1)$ and $B=1+o(1)$. We give a partial negative answer:

\begin{itemize}
\item Doron et al. compose a PEG (pseudo-entropy generator) with an extractor.  The PEG is based on a hardness assumption for MA-type circuits.  We show that black-box PEG constructions from \textsc{the extreme high-end hardness assumption} must have large seed length (and so cannot be used to obtain extreme high-end PRGs by applying an extractor).

To prove this, we establish a new property of (general) black-box PRG constructions from hard functions: it is possible to fix many output bits of the construction while fixing few bits of the hard function. This property distinguishes PRG constructions from typical extractor constructions, and this may explain why it is difficult to design PRG constructions.

\item Chen and Tell compose two PRGs: $G_1:\{0,1\}^{(1+o(1)) \cdot \log m} \to \{0,1\}^{r_2=m^{\Omega(1)}}$ and $G_2:\{0,1\}^{r_2} \to \{0,1\}^m$.  The first PRG is based on \textsc{the extreme high-end hardness assumption}, and the second PRG needs to run in time $m^{1+o(1)}$, and is based on one way functions. We show that in black-box proofs of hardness amplification to $\frac{1}{2}+1/m$, reductions must make $\Omega(m)$ queries, even in the extreme high-end. Known PRG constructions from hard functions are black-box and use (or imply) hardness amplification, and so cannot be used to construct a PRG $G_2$ from \textsc{the extreme high-end hardness assumption}.

The new feature of our hardness amplification result is that it applies even to the extreme high-end setting of parameters, whereas past work does not. Our techniques also improve recent lower bounds of Ron-Zewi, Shaltiel and Varma (ITCS 2021) on the number of queries of local list-decoding algorithms.
\end{itemize}
\end{abstract}
\thispagestyle{empty}
\end{titlepage}


\section{Introduction}

\subsection{Background}
\label{sec:intro:background}

The hardness vs. randomness paradigm (initiated in \cite{Yao82,BM,NW} and followed up by a long sequence of work \cite{BFNW,I95,IW97,STV,KvM,MV99,ISW99,ISWJournal,SU01,Umans02,Umans05,SU05,AIKS16,DMOZ,CT1,CT2}) aims to explicitly construct pseudorandom generators (PRGs) from explicit hard functions.

\begin{definition}[PRGs]
\label{dfn:PRG}
A function $G:\B^r \to \B^m$ is an \remph{$\eps$-PRG} for a function $D:\B^m \ar \B$, if
\[ |\Pr[D(G(U_r))=1]-\Pr[D(U_m)=1]| \le \eps. \]
$G$ is an $\eps$-PRG for a class $\mathcal{D}$ of functions $D:\B^m \to \B$, if for every $D$ in $\mathcal{D}$, $G$ is an $\eps$-PRG for $D$.
If we omit $\eps$ or $\mathcal{D}$, the default choices are $\eps=1/10$, and $\mathcal{D}$ is the class of circuits of size $m$.
\end{definition}

Explicit pseudorandom generators have many applications in computer science. The signature application of PRGs is to derandomize randomized algorithms (by running the algorithm using all outputs of the PRG). This is quantitatively specified in the proposition below.

\begin{proposition}[standard]
\label{thm:PRG=>derandomiztion}
If $G:\B^r \to \B^m$ is a PRG, then every randomized algorithm running in time $m$ can be simulated by a deterministic algorithm in time $\Tit(G) + 2^r \cdot m$, where $\Tit(G)$ is the time it takes to compute the output of $G$ on all $2^r$ inputs, and is obviously upper bounded by $2^r \cdot \Ti(G)$, where $\Ti(G)$ is the time it takes to compute $G$ on a given input.
\end{proposition}

\rparagraph{High-end PRGs that imply BPP=P.}
A corollary of Proposition \ref{thm:PRG=>derandomiztion} is that a PRG $G:\B^{r=O(\log m)} \to \B^m$  with $\Ti(G)=\poly(m)$ implies that BPP=P.\footnote{Note that in this range of parameters there is no reason to distinguish between $\Ti(G)$ and $\Tit(G)$, as $\Tit(G)=\poly(m)$ if and only if $\Ti(G)=\poly(m)$, and this is why past work is stated in terms of $\Ti(G)$ and not $\Tit(G)$.}

Such PRGs are often referred to as \emph{``high-end PRGs"}. Historically, this name aims to distinguish them from weaker \emph{"low-end PRGs''} which have $r=m^{o(1)}$, and $\Ti(G)=2^{O(r)}$, which in turn imply the weaker conclusion that BPP is in subexponential time, see \cite{ISWJournal,SU01} for a discussion.

\rparagraph{Extreme high-end PRGs and fast derandomization.}
Recently, Doron et al. \cite{DMOZ} asked whether it is possible to obtain a faster derandomization. Here, the goal is to show that a randomized algorithm running in time $m$ can be simulated by a deterministic algorithm running in time $O(m^c)$ for the smallest possible constant $c$.

The time of the deterministic simulation of Theorem \ref{thm:PRG=>derandomiztion} depends on both the seed length $r$, and $\Tit(G)$. Note that even if we take $r$ to the extreme\footnote{A PRG must have $r \ge \log m - O( \log \log m)$ as otherwise, a circuit of size $m$ could be hardwired with prefixes of length $r+1$ for all $2^r$ pseudorandom strings, and distinguish a uniform string from a pseudorandom string.}, and have a PRG with $r=1 \cdot \log m$, then the time of the simulation is at least $2^r \cdot m=m^2$. This time can be achieved if furthermore, $\Tit(G)=O(2^r \cdot m)=O(m^2)$ (which follows if $\Ti(G)=O(m)$). This means that we can hope to achieve $c=2$ (that is, a quadratic time simulation) if we have such PRGs, which we will call \emph{``extreme high-end''} PRGs.

\begin{definition}[Extreme high-end PRGs]
\label{dfn:extreme PRGs}
$G:\B^r \to \B^m$ is an \remph{extreme high-end} PRG if:
\begin{description}[noitemsep]
\item{Pseudorandomness:} $G$ is a PRG with seed length $r=(1+o(1)) \cdot \log m$.
\item{Explicitness:} $\Tit(G)=m^{2+o(1)}$ (which follows if $\Ti(G)=m^{1+o(1)}$).
\end{description}
\end{definition}

These parameters are chosen so that an extreme high-end PRG implies that randomized algorithms running in time $m$ can be simulated deterministically in time $m^{2+o(1)}$.

There are reasons to think that this quadratic slowdown is the best possible if one seeks the smallest possible $c$ such that every randomized algorithm running in time $m$ can be simulated deterministic time $m^c$. More precisely, the problem of ``univariate identity testing'' is in BPTIME$(\Tilde{O}(n))$ but not in DTIME$(n^{2-o(1)})$, under certain assumptions on ``fine grained complexity'' introduced by Carmosino et al. \cite{CarmosinoGIMPS16}. See \cite{DMOZ} for details and a discussion.

We remark that this lower bound still allows a deterministic simulation that runs in time $O(m \cdot n)$ (where $n$ is the input length) and a simulation that approaches this time bound was obtained by Chen and Tell \cite{CT1} under certain hardness assumptions. See \cite{CT1} for details and a discussion.

\rparagraph{Hardness implied by PRGs.}
PRGs immediately imply circuit lower bounds that are beyond our current ability.
Consequently, constructing explicit PRGs, requires circuit lower bounds (namely the existence of explicit functions that cannot be computed by small circuits). In particular, high-end PRGs imply the existence of functions $f:\B^{\ell} \to \B$ which cannot be computed by circuits of size $2^{\Omega(\ell)}$, and $\Ti(f)=2^{O(\ell)}$.\footnote{More specifically, Impagliazzo, Shaltiel and Wigderson \cite{ISW99} showed that if $G:\B^r \to \B^m$ is a PRG then for $\ell=r+1$, there is a function $f:\B^\ell \to \B$ (defined by checking whether an input of length $\ell$ is a prefix of an output of $G$) such that $f$ cannot be computed by circuits of size $m$, and $\Tit(f) \approx \Tit(G)$. Note that $\Tit(f)=2^{O(\ell)}$ iff $\Ti(f)=2^{O(\ell)}$.}

\subsubsection{The Impagliazzo-Wigderson (high-end) hardness assumption}
\label{sec:intro:background:IW assumption}

The goal of the hardness vs. randomness program is to construct PRGs based on lower bounds that are as strong (or almost as strong) as the ones implied by the PRG. A major milestone in this program was achieved by Impagliazzo and Wigderson \cite{IW97}.

\begin{theorem}[\cite{IW97}]
\label{thm:IW97}
A high-end PRG follows from the following assumption:
\smallskip
\begin{description}[noitemsep]
\item{} \textsc{The high-end hardness assumption}:  There exist constants $0<\beta< 1< B$, and a function $f:\B^{\ell} \to \B$ that satisfies:
    \begin{description}[noitemsep]
    \item{Hardness:} $f$ cannot be computed by circuits of size $2^{\beta \cdot \ell}$.
    \item{Explicitness:} $\Ti(f) \le 2^{B \cdot \ell}$.
    \end{description}
\end{description}
\end{theorem}

Theorem \ref{thm:IW97} converts hardness into pseudorandomness at close to the ``correct rate'' (as in the converse direction) if one does not care about the precise values of the constants $\beta,B$, and the constant hidden in the $O(\cdot)$ notation in the seed length of the high-end PRG.

\subsubsection{Scaling the Impagliazzo-Wigderson assumption to the extreme high-end}
\label{sec:intro:background:scaled IW assumption}

In the case of extreme high-end PRGs, we insist on seed length $r \approx 1 \cdot \log m$ and the constants $\beta,B$ from the \textsc{The high-end hardness assumption} do matter. Assuming that we don't expect to ``improve'' the hardness of the assumed explicit hard function, we must have $\beta \ge 1-o(1)$ and $B \le 1+o(1)$. Thus, imitating the approach of Impagliazzo and Wigderson \cite{IW97} for the extreme high-end, leads to the following open problem.

\begin{open problem}
\label{open problem:extreme}
Show that an extreme high-end PRG follows from the following assumption:
\begin{description}[noitemsep]
\item{} \textsc{The extreme high-end hardness assumption}: There exists a function $f:\B^{\ell} \to \B$ that satisfies:
    \begin{description}[noitemsep]
    \item{Hardness:} $f$ cannot be computed by circuits of size $2^{(1-o(1)) \cdot \ell}$.
    \item{Explicitness:} $\Ti(f) \le 2^{(1+o(1)) \cdot \ell}$.\footnote{In the case of the extreme high-end, it does make sense to distinguish between $\Ti(f)$ and $\Tit(f)$, and one can consider starting from a weaker explicitness condition in which it is required that $\Tit(f) \le 2^{(2+o(1)) \cdot \ell}$.}
    \end{description}
\end{description}
\end{open problem}

The proof techniques of Impagliazzo and Wigderson \cite{IW97} (as well as of later works \cite{STV,SU01,Umans02}) do not solve Open Problem \ref{open problem:extreme}. As we explain in Section \ref{sec:intro:bb proofs:hybrid}, these proofs rely on the ``hybrid argument'' of \cite{Yao82,GM84}, and even assuming \textsc{the extreme high-end hardness assumption}, one can at best obtain a PRG $G:\B^r \to \B^m$ with seed length $r \ge A \cdot \log m$, where $A > 3$, and actual proofs do worse.

\rparagraph{Recent work on extreme high-end PRGs.}
Recently, Doron et al. \cite{DMOZ}, and Chen and Tell \cite{CT1} gave conditional constructions of extreme high-end PRGs, however, in both cases the assumption used is \emph{stronger} than \textsc{the extreme high-end hardness assumption}. We will elaborate on these results later on.\footnote{Both these papers aim for a slightly weaker goal. Rather than requiring a single PRG with seed length $r=(1+o(1)) \cdot \log m$ and explicitness $m^{1+o(1)}$ as in Definition \ref{dfn:extreme PRGs}, their constructions show that for every $\gamma>0$ there exists a PRG with seed length $r=(1+\gamma) \cdot \log m$ and explicitness $m^{1+\gamma}$. We will not distinguish these two goals in the informal discussion in the introduction.}
An incomparable assumption was very recently used by Chen and Tell \cite{CT2} for constructing ``targeted PRGs'' (which are weaker than PRGs and yet suffice for derandomizing randomized algorithms).

\rparagraph{Goals of this paper.}
In this paper, we investigate the problem of constructing explicit PRGs from explicit hard functions, focusing on open problem \ref{open problem:extreme}. More specifically, we investigate the power of ``black-box proofs'' that convert explicit hard functions into PRGs and related objects. We show limitations on certain recent approaches to solve Open Problem \ref{open problem:extreme}, and hope that this may help to point us in the direction of better constructions. A secondary goal of this paper is to survey recent work and point out the relationship between parameters, and potential barriers for improvement.

\subsection{Black-box proofs}
\label{sec:intro:bb proofs}

\subsubsection{Black-box proofs for PRG constructions and hardness amplification}
\label{sec:intro:bb proofs:PRG}

A black-box proof that converts hard functions into PRGs consists of two parts:
\begin{itemize}[noitemsep]
\item A \emph{construction map}. This is a map that given a candidate function $f:\B^{\ell} \to \B$ produces a candidate PRG $\Con(f)$. (To avoid clutter, we will denote the function $\Con(f)$ by $\Con_f:\B^r \to \B^m$).
\item A \emph{reduction} establishing the correctness of the construction. This is an oracle procedure $\Red^{(\cdot)}$ which given oracle access to an adversary $D$ that breaks the security of $\Con_f$, implements an adversary $C$ that breaks the security of $f$.
\end{itemize}
As we explain in Section \ref{sec:intro:bb proofs:applications}, because of their combinatorial properties, reductions of this type must be \emph{nonuniform} and receive a ``nonuniform advice string'' $\alpha$ (that may depend on the candidate function $f$ and the adversary $D$). This is reflected in the formal definition below.

We will consider ``black-box $\mathsf{A} \Rightarrow \mathsf{B}$ proofs'' for several other choices of primitives $\mathsf{A,B}$ (and not just hard functions and PRGs). One such primitive is functions that are hard on average (meaning that small circuits cannot compute $f$ correctly with high probability on a uniformly chosen input). Such functions can serve both as $\mathsf{A}$ and $\mathsf{B}$ in black-box $\mathsf{A} \Rightarrow \mathsf{B}$ proofs. In order to capture all various scenarios in one definition,
we will use a terminology that will describe a primitive by what it means to ``break the security'' of the primitive.

\begin{definition}
Let $\Func_{n,m}$ denote the set of all functions from $n$ bits to $m$ bits.
\begin{itemize}[noitemsep]
\item For  $G \in \Func_{r,m}$ and $D \in \Func_{m,1}$, we say that $D$ $\eps$-$\PRGB$ $G$, if $G$ is not an $\eps$-PRG for $D$.
\item For $f,C \in \Func_{\ell,1}$, we say that $C$ $\rho$-$\HFB$ $f$, if $\Pr_{x \from U_{\ell}}[C(x)=f(x)] \ge \rho$. We say that $C$ $\HFB$ $f$ if $C$ $1$-$\HFB$ $f$ (meaning that $C=f$). We say that $f$ is a $\rhoHF{\rho}$ for $C$ if $C$ does not $\rho$-$\HFB$ $f$.
    \end{itemize}
\end{definition}

A $\rhoHF{\rho}$ for circuits of a certain size, is an average-case hard function, and the case where $\rho=1$ captures the previously considered notion of worst-case hard functions.

We now formally define black-box $\rhoHFtoepsPRG{\rho}{\eps}$, and $\rhoHFtorhoHF{\rho}{\rho'}$ proofs.

\begin{definition}[Black-box proofs]
\label{dfn:bb}
Given parameters $\ell,r,m,a,\rho,\eps$ (resp. $\ell,\ell',a,\rho,\rho'$) a \remph{black-box $\rhoHFtoepsPRG{\rho}{\eps}$ proof} (resp. a \remph{black-box $\rhoHFtorhoHF{\rho}{\rho'}$ proof}) is a pair $(\Con,\Red)$ of:
\begin{itemize}[noitemsep]
\item A construction map $\Con: \Func_{\ell,1} \to \Func_{r,m}$ (resp. $\Con: \Func_{\ell,1} \to \Func_{\ell',1}$). (We use $\Con_f$ to denote the function $\Con(f)$).
\item An oracle procedure $\Red^{(\cdot)}(x,\alpha)$ such that:
\begin{description}[noitemsep]
\item For every $f \in \Func_{\ell,1}$ and for every $D \in \Func_{m,1}$ such that $D$ $\eps$-$\PRGB$ $\Con_f$ \\ (resp. for every $D \in \Func_{\ell',1}$ such that $D$ $\rho'$-$\HFB$ $\Con_f$),
\item there exists $\alpha \in \B^a$, such that the function $C \in \Func_{\ell,1}$ defined by $C(x)=\Red^D(x,\alpha)$, $\rho$-$\HFB$ $f$.
\end{description}
\end{itemize}
If we omit $\rho$, we mean $\rho=1$. If we omit $\eps$, we mean $\eps=1/10$.
We say that $\Red$ makes $q$ queries, if for every $D \in \Func_{m,1}$, $\alpha \in \B^a$, and $x \in \B^{\ell}$, $\Red^D(x,\alpha)$ makes at most $q$ oracle queries.
\end{definition}

\subsubsection{Parameters of black-box proofs}
\label{sec:intro:bb proofs:parameters}

To the best of our knowledge all hardness vs. randomness proofs of PRG constructions are black-box (or rely on components which are black-box).
In a black-box proof, the advice length $a$ and the number of queries $q$ determine the ``hardness loss'' in the tradeoff. More specifically:

\begin{proposition}[Number of queries determines hardness loss in black-box proofs] $\mbox{}$
\label{prop:bb=>object}
Let $(\Con,\Red)$ be a black-box $\HFtoPRG$ (resp. $\HFtorhoHF{\rho'}$) proof in which $\Red$ makes $q$ queries, and has advice length $a$. If we start from a function $f:\B^{\ell} \to \B$ that cannot be computed by circuits of size $s$, and apply the black-box proof, one can at best obtain that $\Con_f$ is a $\PRG$ (resp. $\rho'$-$\HF$) for circuits of size $m \le \frac{s-a}{q} \le \frac{s}{q}$.
\end{proposition}

Loosely speaking, Proposition \ref{prop:bb=>object} follows because, when measuring the size $s$ of the circuit $C=\Red^D(\cdot,\alpha)$ that is implied by the reduction for a circuit $D$ of size $m$, then this circuit is of size $q \cdot m +a$. This gives that $s \ge q \cdot m +a$, implying the proposition.

Every function $f:\B^{\ell} \to \B$ has circuits of size $2^{\ell}$. This means that in order for reductions to be useful in transforming hard-functions into PRGs (or average case hard functions) they must make $q \le \frac{s-a}{m} \le \frac{2^{\ell}-a}{m}$ queries, and are useless for this purpose, if $q \ge 2^{\ell}$. Moreover, in the extreme high-end, $m=2^{(1-o(1)) \cdot \ell}$ and so it is critical that $q \le \frac{s}{m} \le \frac{2^{\ell}}{m} = 2^{o(\ell)} = m^{o(1)} \ll m$.

\rparagraph{Our notion of black-box does not guarantee explicitness.}
We place no limitation on the map $\Con$, and so, the notion of black-box that we use, \emph{does not} enforce that if $f$ can be computed efficiently, then $\Con_f$ can be computed efficiently. This notion of black-box does not imply the \emph{explicitness} of the constructed function $\Con_f$. We make this choice, because we want to show impossibility results on black-box proofs, and this choice makes our results stronger.\footnote{One way to preserve efficiency is to require that there is an oracle machine $A^{(\cdot)}$ (in some complexity class) such that for every $f \in \Func_{\ell,1}$, $A^f$ implements $\Con_f$. See e.g., \cite{ViolaPH}.} We also remark that the terms "black-box" and ``non-black-box'' are used to mean many different things in the literature.\footnote{For example, sometimes the term ``non-black-box'' is used to denote a deterministic simulation of randomized algorithms that is tailored for the specific input supplied to the algorithm. This notion is used for example in the recent work of Chen and Tell \cite{CT2} which constructs ``targeted-PRGs'' that are targeted to the given input, and is unrelated to the notion of ``black-box'' used here.}

\rparagraph{Parameters for black-box proofs for the extreme high-end.}
As a consequence of Proposition \ref{prop:bb=>object}, if we assume \textsc{the extreme high-end hardness assumption}, to obtain an extreme high-end PRG (as in Open Problem \ref{open problem:extreme}) using a black-box $\HFtoPRG$ proof, we first need to solve the following open problem:

\begin{open problem}
\label{open problem:exist proof?}
Does there exist a black-box $\HFtoPRG$ proof with:
\begin{description}[noitemsep]
\item \textbf{Seed length:} $r=(1+o(1)) \cdot \ell$. (Any black-box proof must have $r \ge \ell$).
\item \textbf{Output length:} $m=2^{(1-o(1))\cdot \ell}$. (Any black-box proof must have $m \le 2^{\ell}$).
\item \textbf{Advice string length:} $a=m^{1+o(1)}=2^{(1-o(1))\cdot \ell}$. (Any black-box proof must have $a \ge m$).
\item \textbf{Number of queries:} $q=m^{o(1)}=2^{o(\ell)}$.
\end{description}
\end{open problem}

We stress again that a positive answer to Open Problem \ref{open problem:exist proof?} is a \emph{necessary} condition for using a black-box proofs to construct an extreme high-end PRG from \textsc{the extreme high-end hardness assumption}, however, it is not a \emph{sufficient} condition.

We do not know whether black-box $\HFtoPRG$ proofs as in Open Problem \ref{open problem:exist proof?} exist.\footnote{In fact, the only known lower bound on black-box $\rhoHFtoepsPRG{\rho}{\eps}$ proofs by Shaltiel, Grinberg and Viola \cite{GSV18} shows that if $a \le 2^{\nu \cdot \ell}$ for some constant $\nu>0$, and $\rho<1-2^{-\ell}$ is sufficiently larger than $\half+\eps$ then $q \ge \Omega(\frac{\log(1/(1-\rho))}{\eps^2})$. In the case of interest where $\rho=1$ and $\eps$ is constant, this gives a weak bound of $q \ge \Omega(\ell)$, and even this does not apply in the extreme high-end where $a=2^{(1-o(1)) \cdot \ell}$. Moreover, if we start from average-case hardness, it is open to prove that $q>1$ for a black-box $\rhoHFtoepsPRG{\rho}{\eps}$ proofs with $\rho \le \half+\eps$, even for small~$a$.}
In Section \ref{sec:intro:results} we show obstacles on certain approaches to design a black-box $\HFtoPRG$ proof, meeting the parameters of Open Problem \ref{open problem:exist proof?}. More specifically, we show that certain approaches cannot yield reductions with few queries. Before we describe our results, we give more background on black-box proofs.

\subsubsection{Black-box proofs as codes, extractors, and other applications}
\label{sec:intro:bb proofs:applications}

Black-boxness is often helpful in PRG constructions (as demonstrated in \cite{NW,KvM,MV99,SU01}) as such proofs readily extend to other computational models (e.g, bounded depth circuits, or nondeterministic circuits). There are other motivations to study black-box proofs (in addition to PRG constructions and hardness amplification). In fact, in the connections and applications below, the ``black-boxness'' of the proofs is crucial and helpful.

\rparagraph{List-decodable codes.} Following Sudan, Trevisan and Vadhan \cite{STV}, a Black-box $\HFtorhoHF{\rho'}$ proof $(\Con,\Red)$ yields a ``list-decodable code'' $E:\B^{2^{\ell}} \to \B^{2^{\ell'}}$, defined by $E(f)_y=\Con_{f}(y)$. A consequence of lower bounds on the rate of such codes is that $\ell' \ge \ell + 2 \cdot \log(1/\rho')$, and that $a \ge 2 \cdot \log(1/\rho')$. Viewing the reduction as a ``list'' of $2^a$ procedures (one for every advice string $\alpha \in \B^a$) yields a variant of a ``local list-decoding algorithm'' for the defined code, with the same number of queries.  Techniques developed for black-box proofs \cite{SV08,GSV18,AASY15} have been useful in proving lower bounds on the number of queries of such codes \cite{RSV21}.

\rparagraph{Randomness Extractors:} Following Trevisan \cite{Tre99}, a black-box $\HFtoepsPRG{\eps}$ proof $(\Con,\Red)$ yields a ``randomness extractor'' $E:\B^{2^{\ell}} \times \B^r \to \B^{m}$, defined by $E(f,y)=\Con_f(y)$, and extracting randomness from sources with min-entropy $k=a+\log(1/\eps)+O(1)$. A consequence of lower bounds on extractors \cite{RT} is that $r \ge \ell + 2 \cdot \log(1/\eps)$, and that $a \ge m+\log(1/\eps) - O(1)$. Continuing the analogy of the previous item, the reduction can be viewed as a local list-decoding algorithm for an ``extractor-code'' \cite{TZ}.  Local list-decoding algorithms for (standard) codes, and for extractor-codes are closely related to ``hard-core bits'' for cryptographic primitives (see e.g. \cite{RSV21} for a discussion).

\rparagraph{Other applications.}
In recent years, black-box $\HFtoPRG$ proofs have found numerous applications in areas that are not directly related to pseudorandomness, and rely on ``black-boxness''. Some examples are: Learning and compression algorithms (Carmosino et al. \cite{CarmosinoEtAl} and subsequent work), worst-case to average-case reductions within NP (Hirahara \cite{Hirahara18} and subsequent work), and Kolmogorov Complexity (Allender et al. \cite{power} and subsequent work).

\subsubsection{Black-box proofs and the hybrid argument}
\label{sec:intro:bb proofs:hybrid}

\paragraph{The hybrid argument cannot be used in the extreme high-end.}
Most known black-box $\HFtoPRG$ proofs in the literature rely on \emph{hardness amplification} in order to use the \emph{hybrid argument} of \cite{Yao82,GM84}. That is, to achieve an $\eps$-PRG, the construction is a sequence of two black-box proofs: $\HF\Rightarrow(\half+\frac{\eps}{m})\mbox{-}\HF\Rightarrow\eps\mbox{-}\PRG$. Thus, even for constant~$\eps$, hardness amplification must be performed to $\rho' \le (\half+\frac{1}{m})$. By the bounds in Section \ref{sec:intro:bb proofs:applications}, in the first step, a function $f_1$ with input length $\ell$ must be transformed into a function $f'$ with input length $\ell' \ge \ell+2 \cdot \log \frac{1}{\rho'} \ge \ell+2\log m$.
The final PRG construction will have seed length $r \ge \ell' \ge \ell+2\log m \ge 3 \log m$. This is too large in the extreme high end, where we want $r=(1+o(1)) \cdot \log m$.

\rparagraph{The hybrid argument suffices for the high-end.}
We remark that taking $\ell'=O(\ell + \log m)$ and $r=O(\ell')$ does suffice for the (non-extreme) high-end, and this is how known constructions for the high-end \cite{IW97,STV,SU01,Umans02} are achieved.\footnote{More specifically, hardness amplification can be performed (by a black-box proof) using ``local list-decodable codes'' \cite{STV}, and the second $\rhoHFtoPRG{(\half+\frac{\eps}{m})}$ step is done using the Nisan-Wigderson generator \cite{NW}, which is a black-box proof. Shaltiel and Umans \cite{SU01} and Umans \cite{Umans02} gave an alternative direct transformation from worst-case hard function into PRGs which achieves a better seed length  in the ``low-end''. However, as it also relies on the hybrid argument (and implies hardness amplification), it cannot achieve the extreme high-end.}

\subsection{Our Results}
\label{sec:intro:results}

We show that certain approaches cannot yield a black-box $\HFtoPRG$ proof with the parameters of Open Problem \ref{open problem:exist proof?}, and therefore cannot be used to solve Open Problem \ref{open problem:extreme}. Our results are summarized in Table \ref{table:PRGs and PEGs} and Table \ref{table:hardness amplification}.

\begin{table}[ht]
\small
\caption{Black-box proofs for PRGs and PEGs}
\begin{centering}
\begin{tabular}{|c|c|c|c|c|}
\hline\hline
Result & Type & Range & Condition  & Bound   \\ [0.5ex] 
\hline 

\hline \cite{GSV18} & $\HFtoepsPRG{\eps}$
 & $a \le 2^{\nu \cdot \ell}$ & & $q \ge \Omega(\frac{\ell}{\eps^2})$ \\ [0.5ex]
\hline Thm \ref{thm:fix} & $\HFtoepsPRG{\eps}$, constant $\eps$
 & $a \le \nu \cdot 2^{\ell}$ & $\exists j =o(\frac{2^{\ell}}{\ell}): \Fix_j(\Con) > a+ j\cdot \ell) $ & $q \ge 2^{\ell}$\\ [0.5ex]
\hline Thm \ref{thm:PEG} & $\HF \Rightarrow \eps$-$\PEG$, constant $\eps$
 & $a \le \nu \cdot 2^{\ell}$ & $r < \ell-\log \ell - O(1)$ & $q \ge 2^{\ell}$ \\ [0.5ex]
\hline \cite{IW97} & $\HFtoepsPRG{\eps}$, constant $\eps$
 & $a \le 2^{\nu \cdot \ell}$ &  & $q \le m^{\Theta(1)}$ \\ [0.5ex]
\hline
\hline
\end{tabular}
\end{centering}
\label{table:PRGs and PEGs}
\end{table}

\begin{table}[ht]
\small
\caption{Black-box proofs for hardness amplification}
\begin{centering}
\begin{tabular}{|c|c|c|c|}
\hline\hline
Result & Type & Range   & Bound   \\ [0.5ex] 
\hline 
\hline \cite{GSV18} & $\HFtorhoHF{(\half+\eps)}$
 & $a \le 2^{\nu \cdot \ell}$ & $q \ge \Omega(\frac{\ell}{\eps^2})$ \\ [0.5ex]
\hline Thm \ref{thm:hardness} &  $\HFtorhoHF{(\half+\eps)}$
 & $a \le \nu \cdot 2^{\ell}$  & $q \ge\Omega(\frac{1}{\eps})$ \\ [0.5ex]
\hline Thm \ref{thm:hardness} & $\HFtorhoHF{(\half+\eps)}$, constant $\eps$
 & $a \le \nu \cdot 2^{\ell}$  & $q \ge \Omega(\ell-\log (2a))$ \\ [0.5ex]
\hline \cite{IW97,STV} & $\HFtorhoHF{(\half+\eps)}$
 & $a \le \nu \cdot 2^{\ell}$ &  $q \le \poly(\frac{\ell}{\eps})$ \\ [0.5ex]
\hline
\hline
\end{tabular}
\end{centering}

In both tables above the first three lines are lower bounds, while the last line is an upper bound, and $0 <\nu \le \half$ is some constant.
\label{table:hardness amplification}
\end{table}

\subsubsection{Limitations on constructions of black-box $\HFtoPRG$ proofs}
\label{sec:intro:results:PRG}

We show that for any black-box $\HFtoPRG$ proof $(\Con,\Red)$, if $\Red$ makes $q \le 2^{\ell}$ queries, then $\Con$ must be structured in a way that allows ``fixing many outputs, with small information cost''. More precisely, we introduce a measure $\Fix_j(\Con)$ defined to be the minimal number $h$, so that when $F$ is chosen at random from $\Func_{\ell,1}$, it is possible to fix $j$ outputs of $\Con_F$, while reducing the information about $F$ by only $h$ bits of information.

\begin{definition}[The cost of fixing $j$ outputs]
\label{dfn:fix}
Given $\Con:\Func_{\ell,1} \to \Func_{r,m}$ we define $\Fix_j(\Con)$ to be the minimal number $h$ such that there exist $j$ distinct outputs $z_1,\ldots,z_j \in \B^m$ such that:
\[ \Pr_{F \from \Func_{\ell,1}}[\forall i \in [j]: \exists y_i \in \B^r \mbox{ s.t. }\Con_F(y_i)=z_i] \ge 2^{-h}. \]
\end{definition}

We show that if $\Red$ makes a small number $q$ of queries, then for every $j$ that is not too large, $\Fix_j(\Con) \le a + j \cdot (\log q+O(1))$. Loosely speaking, this means that after a ``fixed cost'' of $a$ bits of information, a large number of outputs of $\Con_F$ can be fixed at the cost of roughly $\log q$ bits of information about $F$, per $m$-bit output.
This is stated formally below:

\begin{theorem}
\label{thm:fix}
There exists a constant $\nu>0$ such that
for every $\rhoHFtoepsPRG{\rho}{\eps}$ proof $(\Con,\Red)$ with parameters $\ell,r,m,a \le \nu \cdot 2^{\ell},\eps \le 1-2^{r-m},\rho>0.51$, if $\Red$ makes $q \le 2^{\ell}$ queries, then for every $j \le \nu \cdot \frac{2^{\ell}}{\ell}$, \[ \Fix_j(\Con) \le a+j \cdot (\log q + O(1)) \le a+j \cdot (\ell + O(1)) . \]
\end{theorem}

Previous limitations on the number of queries for reductions in black-box proofs (of any type) do not apply when $a \ge 2^{\ell/2}$ and therefore are unapplicable in the extreme high-end

We stress that Theorem \ref{thm:fix} is unrelated to the ``hybrid argument'' and applies even for constructions where the correctness of the reduction \emph{does not} rely on the hybrid argument.
Moreover, the result applies for the whole range of parameters, and regardless of the choices of seed length and output length. See Section \ref{sec:thm:fix} for a more general statement and a discussion.

In the next section we use Theorem \ref{thm:fix} to show limitations on the ``PEG + extractor'' approach of \cite{DMOZ}.

\rparagraph{Distinction between black-box PRGs and typical extractors.}
Following Trevisan \cite{Tre99} (see discussion in Section \ref{sec:intro:bb proofs:applications}) we know that construction maps for black-box $\HFtoPRG$ proofs are extractors (regardless of the number of queries used by the reduction). In fact, extractors and black-box proofs are essentially equivalent if we do not restrict the number of queries made by the reduction.

It is standard that if we choose a construction map $\Con$ at random, it will be an extractor. Nevertheless, we show that it is unlikely that a random construction map $\Con$ will have $\Fix_j(\Con) \le a+j \cdot \ell$ (for relevant values of $j$). This implies that:

\begin{theorem}[informal]
\label{thm:typical informal}
It is unlikely that a random construction map (which is an extractor w.h.p) will have a ``useful'' reduction with $q < 2^{\ell}$.
\end{theorem}

More details, and a precise statement is given in Section \ref{sec:typical extractors}. This demonstrates that requiring a construction to have $q<2^{\ell}$ and be useful for PRGs (and not just for extractors) places limitations on the structure of the construction.

\rparagraph{Our interpretation.}
Our interpretation of Theorem \ref{thm:fix} is that in order to enable the reduction to make few queries, the construction must ``create a backdoor'' and introduce correlations between different outputs. These correlations are ``slightly harmful'' to the goal of being an extractor. More precisely, having low $\Fix_j(\Con)$ means that there is a source distribution with min-entropy that is very high (only lacking $\Fix_j(\Con)$ bits of information) on which $j$ outputs of the extractor are fixed. This will violate the extractor guarantee if $j$ is close to $2^r$ (but is possible for $j \ll 2^r$ which is the case in Theorem \ref{thm:fix}).

Theorems \ref{thm:fix} and \ref{thm:typical informal} suggest that it is more difficult to design PRG constructions than extractors. In Section \ref{sec:revisit} we review the known $\HFtoPRG$ constructions in the literature (the Nisan-Wigderson PRG \cite{NW} and the Shaltiel-Umans PRG \cite{SU01,Umans02}) observing how they achieve low $\Fix_j(\Con)$ in the high-end, and why they do not achieve this in the extreme high-end.
We hope that understanding conditions that $\HFtoPRG$ constructions must satisfy, may point us to new constructions that may be applicable in the extreme high-end.

\rparagraph{Technique.}
We consider a ``distinguisher'' $D_f:\B^m \to \B$ that answers one iff its input is an output of $\Con_f$. For every function $f \in \Func_{\ell,1}$, as $D_f$ $\PRGB$ $\Con_f$, by Definition \ref{dfn:bb}, there must exist $\alpha \in \B^a$ such that the function $C(x)=\Red^{D_f}(x,\alpha)$ satisfies $C=f$. However, $D_f$ only answers one on $2^r$ out of the possible $2^m$ queries. If $\Red$ does not ask such ``interesting queries'', then it obtains no information on $f$, and cannot hope to reconstruct every $f \in \Func_{\ell,1}$.

How does $\Red$ know to ask interesting queries? The advice string $\alpha$ (that depends on $f$) may give $\Red$ information about interesting queries. However, the information in the advice string is limited by its length $a$, and we show that if $\Red$ is able to find interesting queries for many choices of $f \in \Func_{\ell,1}$ and $x \in \B^{\ell}$, then after this ``fixed cost''  of $a$ bits of information, it is still difficult for $\Red$ to find interesting queries, unless the construction $\Con$ is set up so that many interesting queries (that is, outputs of $\Con_f$) have low information, giving that $\Fix_j(\Con) \le a + j \cdot (\log q+O(1))$ for many values of $j$). The precise details are given in Section~\ref{sec:PRGs and PEGs}.

\subsubsection{Limitations on the ``PEG + Extractor'' approach of \cite{DMOZ}}
\label{sec:intro:results:PEG}

Doron et al. \cite{DMOZ} showed how to construct extreme high end PRGs from a strengthening of \textsc{the extreme high-end hardness assumption} of Open Problem \ref{open problem:extreme}. More specifically, rather than only assuming that $f$ cannot be computed by circuits of size $2^{(1-o(1)) \cdot \ell}$, it is assumed that this holds even for circuits that are allowed to use nondeterminism and randomness (and can be thought of as a nonuniform analog of the class MA). This assumption is significantly stronger then \textsc{the extreme high-end hardness assumption} (although, still plausible).

\rparagraph{The PEG + extractor approach.}
The approach of \cite{DMOZ} is to construct a pseudo-entropy generator (PEGs) (for a specific notion of ``computational entropy'' suggested in \cite{BSW03}). This type of PEG can be thought of as a weak notion of PRGs, that is only guaranteed to fool tests that accept a very small fraction of the $2^m$ inputs:

\begin{definition}[PEGs]
\label{dfn:PEG}
A function $G:\B^r \to \B^m$ is a \remph{$(k,\eps)$-PEG} for a function $D:\B^m \ar \B$, if
$\Pr[D(U_m)=1] \le 2^{k-m}$ then $\Pr[D(G(U_r))=1]-\Pr[D(U_m)=1] \le \eps$. We say that $D$ $(k,\eps)$-$\PEGB$ $G$, if $G$ is not an $\eps$-PEG for $D$.\footnote{We remark that the requirement that $\Pr[D(G(U_r))=1]-\Pr[D(U_m)=1] \le \eps$ is sometimes replaced by the stronger requirement that $\Pr[D(G(U_r))=1] \le \eps$, or following \cite{BSW03}, by the requirement that $\Pr[D(G(U_r))=1] \le \Pr[D(U_m)=1] \cdot 2^{m-k}+\eps$ which is stronger still, if we replace $\eps$ by $\eps'=\eps/2$ and $k$ by $k'=k+\log(1/\eps')$.
We are interested in proving limitations on PEG and so taking a weak definition only makes our results stronger (especially as we are interested in constant $\eps$ and $k \ll m$ and the distinction between $k,\eps$ and $k',\eps'$ is immaterial).}

\end{definition}

More specifically, when given a function $f:\B^{\ell} \ar \B$ (which is hard against the stronger model of circuits equipped with nondeterminism and randomness) the construction of \cite{DMOZ} works in two steps:

\begin{enumerate}[noitemsep]
\item{\textsf{PEG:}} Use the hard function to construct a PEG $\operatorname{PEG}:\B^{r_{\operatorname{PEG}}=o(\ell)} \to \B^{m_{\operatorname{PEG}}=2^{(1-o(1)) \cdot \ell}}$ for $k_{\operatorname{PEG}}=2^{(1-o(1)) \cdot \ell}$.
\item{\textsf{Extractor:}} Use an explicit extractor $\operatorname{EXT}:\B^{m_{\operatorname{PEG}}} \times \B^{r_{\operatorname{EXT}}=(1+o(1)) \cdot \ell} \to \B^{m=2^{(1-o(1)) \cdot \ell}}$ with entropy threshold $k_{\operatorname{PEG}}$ (such explicit constructions are known unconditionally).
\end{enumerate}
The final PRG $G:\B^{r=r_{\operatorname{PEG}}+r_{\operatorname{EXT}}} \ar \B^m$ is obtained by interpreting a string $y \in \B^r$ as two strings $y_1 \in \B^{r_{\operatorname{PEG}}}$ and $y_2 \in \B^{r_{\operatorname{EXT}}}$, and setting $G(y)=\operatorname{EXT}(\operatorname{PEG}(y_1),y_2)$.

\rparagraph{The seed length of a PEG.}
The final seed length of $G$ is $r=r_{\operatorname{PEG}}+r_{\operatorname{EXT}}$. Using lower bounds on the seed length of extractors \cite{RT}, it follow that $r_{\operatorname{EXT}} \ge \log(m_{\operatorname{PEG}}-k_{\operatorname{PEG}}) \ge (1-o(1)) \cdot \ell$. Therefore, in order to achieve $r=(1+o(1)) \cdot \ell$ (as is the case in the extreme high-end) it is crucial to use a PEG with seed length $r_{\operatorname{PEG}}=o(\ell)$. (We note that unlike PRGs, PEGs can potentially achieve $r = o(\log m)$, whereas, as noted earlier, PRGs must have $r \ge \log m-O(\log \log m)$).

Summing up, the construction of Doron et al. \cite{DMOZ} relies on the observation that PEGs are weaker objects than PRGs (and are therefore easier to construct) and that PEGs may have seed length that is significantly shorter than PRGs, so that summing the two seed lengths can still yield an almost optimal seed length.

\rparagraph{Impossibility for black-box $\HFtoPEG$  proofs.}
We show an obstacle on this approach when starting from \textsc{the extreme high-end assumption} of Open Problem \ref{open problem:extreme}. More specifically, we show that black-box $\HFtoPEG$ proofs with $r < \ell - \log(\ell)$ that make $q \le 2^{\ell}$ queries, do not exist.

This means that the seed length of each of the two steps must be roughly $\ell$ and so the total seed length of a PEG + extractor must be at least \[ r_{\operatorname{PRG}}=r_{\operatorname{PEG}} + r_{\operatorname{EXT}} \ge (\ell-o(1)) + (\ell-o(1)) = 2\ell-o(1) > (2-o(1)) \cdot \log m,\] showing an obstacle for achieving extreme high-end PRGs with this approach. This is stated formally in the next theorem.

\begin{theorem}[Impossibility of black-box PEGs with $r<\log m$]
\label{thm:PEG}
There exists a constant $\nu>0$ such that for every
black-box $\rhoHFtokepsPEG{\rho}{k}{\eps}$ proof $(\Con,\Red)$ with parameters $\ell,r<k,m,a \le \nu \cdot 2^{\ell},\eps \le 1-2^{r-m},\rho \ge 0.51$ such that $\Red$  makes $q \le 2^{\ell}$ queries, it follows that:\footnote{We have not yet formally defined the notion of $\rhoHFtokepsPEG{\rho}{k}{\eps}$ proof. However, this definition is obtained by simply replacing ``PRG-break'' with ``PEG-break'' in Definition \ref{dfn:bb}. For completeness, we give the full definition in Section \ref{sec:thm:PEG}.}
 \[r \ge  \ell - \log \ell - O(1).\]
\end{theorem}

Summing up, Theorem \ref{thm:PEG} shows that black-box proofs cannot be used to solve Open Problem~\ref{open problem:extreme} using the PEG + extractor approach of Doron et al. \cite{DMOZ}.\footnote{In light of Theorem \ref{thm:PEG} one may ask how Doron et al. \cite{DMOZ} construct their PEG. Is their proof non-black-box? The answer is that their proof is black-box, but it allows the reduction $\Red$ to use nondeterminism and randomness (and it is this ability that enables the reduction to make few queries). The cost of using these resources is that the reduction only contradicts the hardness of $f$ if it is assumed to be hard even for circuits equipped with these resources. See e.g., Applebaum et al. \cite{AASY15} for a discussion on nondeterministic reductions.}

\rparagraph{Consequences of Theorem \ref{thm:PEG} for ``quantified derandomization''.}
The notion of PEGs in Definition \ref{dfn:PEG} is closely related to ``quantified derandomization'' (introduced by Goldreich and Wigderson \cite{GWquantified}, see survey article by Tell \cite{TellSurvey}). Quantified derandomization is concerned with derandomizing algorithms that err on very few (say less than $2^k$) of the possible $2^m$ values of their $m$ random bits. This means that PEGs are exactly the type of PRGs that are suitable for this derandomization (see \cite{DMOZ,TellSurvey} for a discussion).

Consequently, Theorem \ref{thm:PEG} can also be viewed as a limitation on black-box proofs that obtain PRGs with very short seed for quantified derandomization, starting from \textsc{the extreme high-end hardness assumption}.

\rparagraph{Technique.}
Theorem \ref{thm:PEG} follows from Theorem \ref{thm:fix} (noting that Theorem \ref{thm:fix} also applies to PEGs). Loosely speaking, if $r$ is small, then the number of outputs of $\Con$ is small, and it is impossible for $\Fix_{j}(\Con)$ to be small for large values of $j$, ruling out black-box proofs in which $r$ is small.

\subsubsection{Lower bounds on black-box hardness amplification at the extreme high-end}
\label{sec:intro:results:hardness}

Grinberg, Shaltiel and Viola \cite{GSV18} (continuing a line of previous work \cite{ViolaThesis,SV08,AS11}) proved a lower bound of $q \ge \Omega(\frac{\ell}{\eps^2})$ on the number of queries in reductions for black-box $\HFtorhoHF{(\half+\eps)}$ proofs (a.k.a. hardness amplification). By Proposition \ref{prop:bb=>object}, such bounds imply that using black-box proofs to convert a function $f$ on $\ell$ bits, that cannot be computed by size $s$ into one that is average case hard for circuits of size $m$, one must have $m \le \frac{s-a}{q} \le \frac{s}{q}$ which means that such transformation ``lose a factor $q$ in the hardness''.

In this paper we prove a lower bound of $q \ge \Omega(\frac{1}{\eps})$, which is quantitatively weaker than that of \cite{GSV18}, but unlike \cite{GSV18} it applies in the extreme high-end. That is, our result allows $a=2^{(1-o(1)) \cdot \ell}$ whereas \cite{GSV18} (as well as all previous bounds) only works if $a \le 2^{\nu \cdot \ell}$ for some constant $\nu>0$. (It is open to match the bound of \cite{GSV18} for large $a$).

\begin{theorem}[Lower bounds on black-box hardness amplification at the extreme high-end]
\label{thm:hardness}
Let $(\Con,\Red)$ be a $\HFtorhoHF{(\half+\eps)}$ proof
with parameters $\ell,\ell',a,\rho=1,\rho'=\half+\eps$. If  $\eps \le \frac{1}{10}$, $\ell' \ge \log(1/\eps) + \Omega(1)$ and $a \le \frac{2^{\ell}}{10}$ then $\Red$ must make at least $q$ queries for \[ q\ge \max\left(\Omega\left(\frac{1}{\eps}\right),\Omega\left(\ell-\log(2a)\right)\right). \]
\end{theorem}

To the best of our knowledge, Theorem \ref{thm:hardness} is the first bound on the number of queries in black-box hardness amplification that applies for $a \ge 2^{\ell/2}$ and to the extreme high-end.

Using Proposition \ref{prop:bb=>object}, Theorem \ref{thm:hardness} implies that even if one starts from \textsc{the extreme high-end hardness assumption}, then to obtain a $\rhoHF{(\half+\frac{1}{m})}$ for circuits of size $m$ (and apply the hybrid argument as explained in Section \ref{sec:intro:bb proofs:hybrid}) there must be a ``hardness loss'', and $m \le \frac{2^{\ell}}{q} \le \frac{2^{\ell}}{m}$, implying that $m \le 2^{\ell/2}$.

Note that this limitation applies regardless of the length $\ell'$ of the input length of $\Con_f$. This means that a black-box $\HFtoPRG$ proof that relies on hardness amplification and the hybrid argument (that is: $\HFtorhoHF{(\half+\frac{1}{m})} \Rightarrow \PRG$) must have $m \le 2^{\ell/2}$, and this holds even if the seed length $r$ of the PRG is large.

In the next section we show that a similar argument also gives limitations on using hardness amplification together with the ``PRG composition'' approach of Chen and Tell \cite{CT1}.

\rparagraph{Technique.} The work of \cite{GSV18} (as well as previous work in this area) relied on information theoretic techniques that break down for large $a$. Indeed, the proof of Theorem \ref{thm:hardness} uses a different argument. This argument builds on ideas of Applebaum et al. \cite{AASY15} which connect the number of queries required by a reduction (or in the case of \cite{RSV21} a local list-decoding algorithm) to the success that small size, constant depth circuits have in solving the ``coin problem'' (that is distinguishing a sequence of independent tosses of an unbiased coin from a sequence of independent tosses of a slightly biased coin). The proofs of \cite{AASY15,RSV21} do not work for $a \ge 2^{\ell/2}$ and our results are obtained by proving tighter bounds on depth 3 circuits for specific versions of the coin problem that come up in the argument. The proof is given in Section \ref{sec:hardness amplification}. As a consequence, we also improve the bounds of \cite{RSV21} on the number of queries of local list-decoding algorithms, see Section \ref{sec:codes} for details.

\subsubsection{The hybrid argument and the ``PRG composition'' approach of \cite{CT1}}
\label{sec:intro:results:HF}

Chen and Tell \cite{CT1} construct extreme high end PRGs if, in addition to \textsc{the extreme high-end hardness assumption} of Open Problem \ref{open problem:extreme}, one also assumes the existence of one-way functions (OWFs). The existence of OWFs is a standard and widely believed assumption in cryptography. Nevertheless, OWFs (or more generally cryptography) are not known to be implied by extreme high-end PRGs (or other PRGs in complexity theory). Assuming OWFs does not seem necessary.

\rparagraph{The PRG composition approach.}
Chen and Tell \cite{CT1} start from a hard function $f:\B^{\ell} \ar \B$ given by \textsc{the extreme high-end hardness assumption}. Their PRG $G:\B^{(1+o(1)) \cdot \ell} \to \B^{2^{(1-o(1)) \cdot \ell}}$ is obtained by PRG composition, namely $G(y)=G_2(G_1(y))$ where:
\begin{enumerate}[noitemsep]
\item $G_1:\B^{(1+o(1)) \cdot \ell} \to \B^{m_1=2^{\Omega(\ell)}}$ is a PRG against circuits of size $2^{(1-o(1)) \cdot \ell}$ that is constructed from \textsc{the extreme high-end hardness assumption} using hardness amplification, the Nisan-Wigderson PRG, and the hybrid argument.\footnote{More precisely, the cost of the hybrid argument (explained in Section \ref{sec:intro:bb proofs:hybrid}) is measured in terms of the output length $m$ (even if the PRG fools circuits of larger size, as is the case here). This means, that the goal of fooling circuits of size $2^{(1-o(1)) \cdot \ell}$ can be achieved by known black-box proofs (in the same manner explained in Section \ref{sec:intro:bb proofs:hybrid}) from the \textsc{extreme high-end hardness assumption} for PRGs that output $m_1=2^{\Omega(\ell)}$ bits, rather than $m=2^{(1-o(1)) \cdot \ell}$ bits.}
\item $G_2:\B^{m_1=2^{\Omega(\ell)}} \to \B^{m=2^{(1-o(1)) \cdot \ell}}$ is a PRG with modest stretch (polynomial rather than exponential) that suffices to push the output length from $m_1=2^{\Omega(\ell)}$ to $m=2^{(1-o(1)) \cdot \ell}$. Nevertheless, for the composition to be a PRG, it is crucial that $G_2$ can be computed in time $2^{(1-o(1)) \cdot \ell}$ (that is in almost linear time in its output length $m$).\footnote{This requirement is necessary as in the composition one needs to consider a distinguisher for $G_1$ that runs $G_2$ as a procedure, and $G_1$ cannot fool circuits of size larger than $2^{\ell}$. We also remark that in contrast to cryptography, in hardness vs. randomness, efficiency of components is rarely used in proving security of the final primitive, and this is one such rare instance.} Such PRGs indeed follow from the existence of OWFs \cite{HILL}.
\end{enumerate}

A natural question is whether it is possible to construct the PRG $G_2$ from \textsc{the extreme high-end hardness assumption}. This will remove the need for OWFs.

PRGs with polynomial stretch follow from this assumption (and even from weaker versions like \textsc{the high-end hardness assumption} or ``low-end'' versions). This is good news, as it shows that hardness amplification and the hybrid argument \emph{can} yield sufficient stretch in this case.

The issue is that the PRGs constructed by these methods do not run in time that is nearly linear in their output length $m$ (and instead run in time $m^c$ where $c>2$). This means that they are unsuitable for the PRG composition approach (and this is why \cite{CT1} relies on OWFs).

\rparagraph{Our results.}

Theorem \ref{thm:hardness} implies that PRGs (even with modest stretch) cannot run in nearly linear time, if they are obtained using black-box hardness amplification and the hybrid argument, assuming \textsc{the extreme high-end hardness assumption}. This implies that if we only assume \textsc{the extreme high-end hardness assumption} (and do not assume the existence of OWFs) then current techniques cannot yield the PRG $G_2$ required by the PRG composition.

More precisely, we have already seen in Section \ref{sec:intro:results:hardness} that in order to do a hybrid argument for output length $m$, one needs a hardness amplification result that amplifies to below $\half+\frac{1}{m}$, and that Theorem \ref{thm:hardness} implies that $m \le 2^{\frac{\ell}{2}}$. This holds in the extreme high-end, regardless of the relationship between $r$ and $m$. In particular, it also holds when trying to construct PRGs with modest stretch like $G_2$, assuming \textsc{the extreme high-end hardness assumption}. On the other hand, assuming that computing the average-case hard function $\Con_f$ takes at least the time it takes to compute the worst-case hard function $f$, and recalling that $f$ cannot be computed by circuits of size $2^{(1-o(1)) \cdot \ell}$, we conclude that $\Con_f$ cannot be computed in time $2^{(1-o(1)) \cdot \ell}$ (which is at least $m^{2-o(1)}$).

Summing up, after performing hardness amplification, there must be at least a quadratic gap between the time it takes to compute $\Con_f$, and the circuit size for which it is hard on average. This gap is inherited by the final PRG $G_2$. Consequently, $G_2$ cannot run in time smaller than $m^{2-o(1)}$, and in particular, there is in obstacle for obtaining PRGs that run in time nearly linear in $m$, using these techniques (even if the stretch is modest).

This shows an obstacle for using current techniques (that rely on hardness amplification and the hybrid argument) to apply the PRG composition approach of \cite{CT1} assuming only \textsc{the extreme high-end hardness assumption}. This partially explains why \cite{CT1} need the additional assumption that OWFs exist in order to construct the PRG $G_2$.


\subsection{Organization of this paper}
In Section \ref{sec:prelims} we define some notation, and cite some previous work that we use. In Section \ref{sec:PRGs and PEGs} we prove our results on black-box proofs for PRGs and PEGs (Theorem \ref{thm:fix} and Theorem \ref{thm:PEG}).
In Section \ref{sec:hardness amplification} we prove our results on hardness amplification (Theorem \ref{thm:hardness}).
In Section \ref{sec:codes} we use the methodology devised for Theorem \ref{thm:hardness} to improve the lower bounds of Ron-Zewi, Shaltiel and Varma \cite{RSV21} on the number of queries of decoders for locally decodable codes. In Section \ref{sec:open problems} we mention some open problems.

\section{Preliminaries}
\label{sec:prelims}

\paragraph{Distributions and Random Variables.}

We use $X \from D$ to denote the experiment in which $X$ is chosen from distribution $D$. For a set $A$ we use $X \from A$ to denote the experiment in which $X$ is chosen uniformly from $A$. Two distributions $X,Y$ over the same finite domain are $\eps$-close if for every event $A$, $|\Pr[X \in A]-\Pr[Y \in A]| \le \eps$.
For a distribution $X$ over $\B^n$, we define $\Hi(X)=\min_{x \in \B^n} \log \frac{1}{\Pr[X=x]}$.

For $0 \le p \le 1$, we use $U^p_n$ to denote the distribution of $n$ i.i.d. random variables, where each one has probability $p$ to be one. We use $U_n$ to denote $U^{1/2}_n$ (the uniform distributions on $n$ bit strings).

\rparagraph{Hamming distance and weight.}
For two strings $x,y \in \B^n$ we use $\mathsf{dist}(x,y)$ to denote the relative Hamming distance between $x$ and $y$, namely, $\mathsf{dist}(x,y)=|\set{i \in [n]:x_i \ne y_i}|/n$.
We use $\wt(x)$ to denote the absolute Hamming weight, namely $\wt(x)=|\set{i \in [n]:x_i \ne y_i}|$.

\paragraph{Restrictions.}
We use the standard notion of random restrictions. Namely, a restriction $\rho:[n] \to \set{0,1,*}$ restricts some of the variables of a function $f:\B^n \to \B$.

We use $\Select(\rho)$ to denote the subset of free variables, $\Free(\rho)$ to denote the number of free variables, and $\Assign(\rho)$ to be the value of the assigned variables. This is stated formally in the next definition.

\begin{definition}[Restrictions]
\label{dfn:restrictions}
A \remph{restriction} of $n$ variables is a function $\rho:[n] \to \set{0,1,*}$. We define $\Select(\rho)=\set{i: \rho(i)=*}$, $\Free(\rho)=|\Select(\rho)|$, and $\Assign(\rho) \in \B^{n-\Free(\rho)}$ by enumerating the fixed elements $(i_1 < \ldots < i_{n-\Free(\rho)})=[n] \setminus \Select(\rho)$ and defining $\Assign(\rho)_j = \rho(i_j)$.

Given a function $f:\B^n \to \B$ and a restriction $\rho:[n] \to \set{0,1,*}$, and $x \in \B^{\Free(\rho)}$ we define $\Fill_{\rho}(x) \in \B^n$ as follows: Let $i_1<\ldots<i_{\Free(\rho)}$ be the indices on which $\rho$ outputs `$*$',
\[ \Fill_{\rho}(x)_i = \left\{\begin{array}{ll}
        x_j , & \exists j  \mbox{ s.t.\ } i=i_j \\
        \rho(i), & \mbox{otherwise}
        \end{array}\right. \]
Given $f:\B^n \ar \B$, we define $f_\rho:\B^{\Free(\rho)} \to \B$ by
\[ f_{\rho}(x)=f(\Fill_{\rho}(x)). \]
\end{definition}

We use $\Rest^n_p$ to denote the set of restrictions with $p \cdot n$ unrestricted variables. A formal definition is given below.

\begin{definition}[The class $\Rest^n_p$]
\label{dfn:Rest}
For $0 \le p \le 1$, let $\Rest^n_p$ denote the set of restrictions $\rho:[n] \to  \set{0,1,*}$ with $\Free(\rho)=p \cdot n$.
\end{definition}

We use the following switching lemma due to Hastad \cite{Hastad}. The specific version used here, appears in Beame's primer \cite{BeamePrimer}.

\begin{theorem} \cite{Hastad,BeamePrimer}
\label{thm:switching}
Let $C$ be a $q$-CNF over $n$ variables. For every $0 \le p \le 1$, the probability over choosing $\rho \from \Rest^n_p$ that $C_{\rho}$ does not have a decision tree of height $h$ is at most $(7 \cdot p \cdot q)^h$.
\end{theorem}

\section{Limitations on black-box proofs for PRGs and PEGs}
\label{sec:PRGs and PEGs}

In this section we discuss our main results on black-box proofs for PRGs and PEGs (Theorem \ref{thm:fix} and Theorem \ref{thm:PEG}). In Section \ref{sec:thm:fix} we discuss Theorem \ref{thm:fix}, and in Section \ref{sec:thm:PEG} we discuss Theorem \ref{thm:PEG}. The proofs of both these Theorems is given in Section \ref{sec:prf:fix and PEG}.

\subsection{Limitations on constructions of black-box $\HFtoPRG$ proofs}
\label{sec:thm:fix}

In this section we restate Theorem \ref{thm:fix} and discuss its consequences and interpretation. In Section \ref{sec:thm:fix restate} we restate Theorem \ref{thm:fix} and discuss its interpretation.  In Section \ref{sec:typical extractors} we show that a corollary of Theorem \ref{thm:fix} is that PRG constructions must differ from typical extractor constructions, giving a precise statement of Theorem \ref{thm:typical informal}. In Section \ref{sec:revisit} revisit the Nisan-Wigderson PRG \cite{NW} and the Shaltiel-Umans PRG \cite{SU01,Umans02} and discuss them from the perspective of Theorem \ref{thm:fix}.

\subsubsection{A general statement of Theorem \ref{thm:fix}}
\label{sec:thm:fix restate}

\paragraph{Review of the setup of Theorem \ref{thm:fix}.}
Theorem \ref{thm:fix} shows that for any black-box $\HFtoPRG$ proof $(\Con,\Red)$, if $\Red$ is useful, and makes $q \le 2^{\ell}$ queries, then $\Con$ must be structured in a way that allows ``fixing many outputs, with small information cost''. This is measured by $\Fix_j(\Con)$ from Definition \ref{dfn:fix}, which defines $\Fix_j(\Con)$ to be the minimal number $h$, so that when $F$ is chosen at random from $\Func_{\ell,1}$, it is possible to fix $j$ outputs of $\Con_F$, while only reducing $h$ bits of information about $F$.

Theorem \ref{thm:fix} shows that in order for $\Red$ to be useful and have $q \le 2^{\ell}$, it must be that for many sufficiently large choices of $j$, $\Fix_j(\Con) \le a + j \cdot (\log q +O(1))$. This means that after a ``fixed cost'' of $a$ bits of information, a large number of outputs of $\Con_F$ can be fixed at the cost of $\log q+O(1)$ bits of information about $F$, per output.

\paragraph{A more general statement.}
The following theorem is a generalized version of Theorem \ref{thm:fix} which also allows the parameter $\rho$ (measuring how hard on average is the function we start from) to be very close to $\half$, and then the amortized cost $\log q + O(1)$, is replaced by $\log \frac{q}{\rho-\half}$.

\begin{theorem}
\label{thm:fix general}
Let $(\Con,\Red)$ be a black-box $\rhoHFtoepsPRG{\rho}{\eps}$ proof for parameters $\ell,r,m,a,\eps,\rho$ such that $\Red$ makes at most $q \le 2^{\ell}$ queries. If $\rho=\half+\eta$, $\eta \ge 2^{-\ell}$, $\eps \le 1-2^{r-m}$, and $a \le \nu \cdot \eta^2 \cdot 2^{\ell}$ for some sufficiently small constant $\nu>0$, then for $j_{\max} = \nu \cdot \frac{\eta^2 \cdot 2^{\ell}}{\ell}$, and every $j \le j_{\max}$, \[\Fix_j(\Con) \le a + j \cdot (\log q + \log \frac{4}{\eta}). \]
\end{theorem}

Theorem \ref{thm:fix} immediately follows from Theorem \ref{thm:fix general} by setting $\rho=0.51$. Theorem \ref{thm:fix general} is proven in Section \ref{sec:prf:fix and PEG}.

\paragraph{A discussion of the parameters of Theorem \ref{thm:fix general}.}
The parameters in Theorems \ref{thm:fix} and Theorem \ref{thm:fix general} are applicable even for very large $a$, and therefore Theorem \ref{thm:fix general} applies in the extreme high-end (where $a=2^{(1-o(1)) \cdot \ell}$) as well as in less challenging ranges such as the high-end (where $a=2^{\nu \cdot \ell}$ for a constant $\nu>0$) and the low-end (where $a=\poly(\ell)$). We stress that to the best of our knowledge, previous limitations on the number of queries of a black-box reduction (of any kind) did not apply for $a \ge 2^{\ell/2}$ and in particular, for the extreme high-end. See Section \ref{sec:hardness amplification} for a discussion on past lower bounds on the number of queries by reductions for black-box hardness amplification proofs.

Furthermore, Theorems \ref{thm:fix} and Theorem \ref{thm:fix general} make no assumption on the stretch (namely, the relationship between $r$ and $m$) of the constructed PRG, and apply even when the stretch is very small, and $m$ is only slightly larger than $r$. In addition, the theorem applies even when $\eps$ is very large, and approaches one (rather than zero).

Finally, Theorem \ref{thm:fix general} applies even when $\eta=\half-\rho$ is very small (say $\eta=2^{-\Omega(\ell)})$ and even for black-box $\HFtoPRG$ proofs like the Nisan-Wigderson PRG, which are only known to work when $\eta$ is very small. See Section \ref{sec:revisit} for a discussion of the Nisan-Wigderson PRG.

\subsubsection{black-box PRGs are different than typical extractors}
\label{sec:typical extractors}

In this section we prove Theorem \ref{thm:typical informal} which loosely states that a random construction $\Con$ for a black-box $\HFtoPRG$ proof is likely to be an extractor, but unlikely to be useful when transforming hard functions into PRGs.

\paragraph{Formal connection between extractors and black-box $\HFtoPRG$ proofs.}
It is well known following Trevisan's breakthrough construction \cite{Tre99} (and as explained in Section \ref{sec:intro:bb proofs:applications}) that black-box $\HFtoPRG$ proofs, are closely related to randomness extractors. Let us formally specify this connection. We start with a formal definition of randomness extractors.

\begin{definition}[extractors]
A function $E:\B^n \times \B^r \to \B^m$ is a $(k,\eps)$-\remph{extractor} if for every distribution $X$ over $\B^n$, with $\Hi(X)\ge k$, $E(X,U_r)$ is $\eps$-close to $U_m$.
\end{definition}

In order to compare function $E:\B^n \times \B^r \to \B^m$ to functions $\Con:\Func_{\ell,1} \to \Func_{r,m}$ we use the following notation.

\begin{definition}[Constructions and extractors]
A string $f \in \B^n$ can be viewed as a function $f:\B^{\log n} \to \B$ by $f(i)=f_i$. This also applies in the other direction, allowing a function $f:\B^{\ell} \to \B$ to be viewed as a string $f \in \B^{2^{\ell}}$.

Given a function $E:\B^n \times \B^r \to \B^m$, we define $\ell=\log n$, and
$\Con^E:\Func_{\ell,1} \to \Func_{r,m}$ where $\Con^E_f:\B^r \to \B^m$ is defined by $\Con^E_f(y)=E(f,y)$. This also applies in the other direction, where a function $\Con:\Func_{\ell,1} \to \Func_{r,m}$ induces a function $E:\B^{n=2^{\ell}} \times \B^r \to \B^m$ by $E(f,y)=\Con_f(y)$.
\end{definition}

With this notation, extractors and black-box $\HFtoPRG$ proofs, are essentially equivalent. This is stated in the following standard proposition which shows that the two notions are roughly equivalent for $a \approx k$.

\begin{proposition}[Extractors are essentially equivalent to black-box $\HFtoPRG$ proofs]
\label{prop:extractors and black-box proofs}
Let $E:\B^n \times \B^r \to \B^m$, and let $\ell=\log n$.
\begin{enumerate}
\item If $\Con^E$ can be matched with an oracle procedure $\Red^{(\cdot)}$ such that the pair $(\Con^E,\Red)$ is a black-box $\HFtoepsPRG{\eps}$ proof with advice length $a$, then $E$ is a $(k,2\eps)$-extractor, for $k=a + \log(1/\eps) + 1$.
\item If $E$ is a $(k,\eps)$-extractor then there exists an oracle procedure $\Red^{(\cdot)}$ such that the pair $(\Con^E,\Red)$ is a black-box $\HFtoepsPRG{\eps}$ proof with advice length $a=k+1$.
\end{enumerate}
\end{proposition}

\paragraph{Reductions must make few queries to be useful.}
Note however, that in order to be useful for PRG construction, a black-box $\HFtoPRG$ proof must have a reduction $\Red$ that make few queries. In particular, as can be seen in Proposition \ref{prop:bb=>object}, the reduction is useless if it makes $q \ge 2^{\ell}$ queries (and even if it makes $q \ge \frac{2^{\ell}-a}{m}$ queries).

The equivalence of Proposition \ref{prop:extractors and black-box proofs} does not mention the number of queries (which can be as large as $2^m$). In other words, a black-box $\HFtoPRG$ proof is an extractor even when the number of queries is very large, but it must have $q \ll 2^{\ell}$ in order to be useful when converting hard functions into PRGs.

\paragraph{A typical extractor is not a useful PRG construction.}
We now show that for a wide range of parameters $\ell,r,m$ and $a$ (and in particular for the parameters that correspond to the high-end and the extreme high-end) if we choose a  function $E:\B^n \times \B^r \to \B^m$ uniformly at random then:
\begin{itemize}
\item It is unlikely that construction $\Con^E$ (associated with $E$) has small $\Fix_j(\Con^E)$, and therefore, by Theorem \ref{thm:fix}, it is unlikely that $\Con^E$ can be matched with a reduction $\Red$ that makes $q \le 2^{\ell}$ queries.
\item However, a standard argument shows that it is likely that $E$ is an extractor, and therefore, by Proposition \ref{prop:extractors and black-box proofs}, it is likely that $\Con^E$ can be matched with a reduction $\Red$ (although $\Red$ may make $q=2^m$ queries).
\end{itemize}
As explained in Section \ref{sec:intro:results:PRG}, this can be interpreted as saying that PRG constructions must be quite different than typical extractor constructions in order to allow the existence of a reduction that makes few queries.

A formal and quantitative version of this is stated in the next definition and theorem.

\begin{definition}[Random construction]
Let $E \from \RndE_{\ell,r,m}$ denote the experiment in which a function $E:\B^{n=2^{\ell}} \times \B^r \to \B^m$ is chosen uniformly from the set of all such functions.
\end{definition}

\begin{theorem}
\label{thm:typical extractors}
Let $\ell,r,m,a,\eps$ be parameters such that $m \ge 2r$ and the following holds:
\begin{description}[noitemsep]
\item{\textsf{Parameters that allow $E$ to be an extractor:}} $r \ge \ell + 2\log(1/\eps) + O(1)$, $2^{\ell} \ge a \ge m+2\log(1/\eps) + O(1)$, and $\eps \le \half$.
\item{\textsf{Parameters that prevent fixing seeds:}}
$a+j \cdot \ell \le min(j \cdot (m-r) -O(1),\frac{2^{\ell}}{2})$.
\end{description}
\noindent
In the experiment $E \from \RndE_{\ell,r,m}$ the following items hold:
\begin{itemize}
\item For every $j >4$, the probability that $\Fix_j(\Con^E) \le a+j \cdot \ell$ is smaller than $2^{-2^{2^{\ell}/2}}$.

Consequently (by Theorem \ref{thm:fix}) the probability that $\Con^E$ can be matched with an oracle procedure $\Red^{(\cdot)}$ that makes $q<2^{\ell}$ queries, is smaller than $2^{-2^{2^{\ell}/2}}$.

\item The probability that $E$ is an $(a+O(1),\eps)$ extractor is larger than $1-o(1)$.
Consequently (by Proposition \ref{prop:extractors and black-box proofs}) the probability that $\Con^E$ can be matched with an oracle procedure $\Red^{(\cdot)}$ with advice length $a+O(1)$ is larger than $1-o(1)$.
\end{itemize}
\end{theorem}

\paragraph{The parameters in Theorem \ref{thm:typical extractors}.}
We stress that Theorem \ref{thm:typical extractors} applies for a wide range of parameters including the extreme high-end, and the high-end. Furthermore, Theorem \ref{thm:typical extractors} does not require that $m$ is much larger than $r$ and applies even when the stretch is small. Furthermore, the probability in the first item is $2^{-2^{2^{\ell}/2}}$ which is quite small.
We also note that the requirement that $a+j \cdot \ell$ is small are very mild in the following sense:
\begin{itemize}[noitemsep]
\item For every $j$ it is obvious that $\Fix_j(\Con) \le j \cdot m$, as the amount of information in $j$ outputs is at most $j \cdot m$. We are only requiring that $a+j \cdot \ell \le j \cdot (m-r)-O(1)$ which is quite close to the obvious bound.
\item The total amount of information in $f$ is at most $2^{\ell}$. We are only requiring that $a+j \cdot \ell \le \frac{2^{\ell}}{2}$ which is quite close to the obvious bound.
\end{itemize}

\begin{proof} (of Theorem \ref{thm:typical extractors})
The second item is a standard calculation which can be found for example in \cite{RT}. Therefore, we will only prove the first item.
Within this proof, unless otherwise specified, all probabilities are regarding the experiment $E \from \RndE_{\ell,r,m}$.

For a fixed set $S \subseteq \B^n$ of size $K=2^{n-(a+j \cdot \ell)}$ and a sequence $\bar{z}=(z_1,\ldots,z_j)$ of distinct elements in $\B^m$, let $A_{S,\bar{z}}$ denote the event that for every $f \in S$, there exist $y_1,\ldots,y_j \in \B^r$ such that for all $i \in [j]$, $E(f,y_i)=z_i$.
This definition is made so that:
\[ \Pr[\Fix_j(\Con^E) \le a + j \cdot \ell] \le {2^n \choose K} \cdot 2^{m \cdot j} \cdot \Pr[A_{S,\bar{z}}],\]
where the latter expression is obtained by a union bound over all choices of $S$ and $\bar{z}$. For every such $S$ and $\bar{z}$, and for every $f \in S$, let $A_{S,\bar{z},f}$ denote the event that there exist $y_1,\ldots,y_j \in \B^r$ such that for all $i \in [j]$, $E(f,y_i)=z_i$. For every $f \in S$, by a union bound over all choices of $y_1,\ldots,y_j \in \B^r$,
\[ \Pr[A_{S,\bar{z},f}] \le 2^{r \cdot j} \cdot 2^{- m \cdot j }. \]
These events are independent for the $K$ different choices of $f \in S$, and therefore, the probability of their conjunction is the product of their individual probabilities.
\[ \Pr[A_{S,\bar{z}}] \le \left(2^{r \cdot j} \cdot 2^{- m \cdot j } \right)^K = 2^{-j \cdot K \cdot (m-r)} . \]
We conclude that:
\begin{align*}
\Pr[\Fix_j(\Con^E) \le a + j \cdot \ell] &\le {2^n \choose K} \cdot 2^{m \cdot j} \cdot 2^{-j \cdot K \cdot (m-r)}  \\
& \le \left(\frac{2^n \cdot e}{K}\right)^K \cdot 2^{m \cdot j} \cdot
 2^{-j \cdot K \cdot (m-r)}
\end{align*}
Let $v=2^{2^{\ell}/2}$. In order to show that this probability is smaller than $2^{-v}$ it is sufficient to show that:
\begin{enumerate}
\item $\left(\frac{2^n \cdot e}{K}\right)^K \cdot
 2^{-j \cdot K \cdot (m-r)} \le 2^{-2v}$, and
\item $2^{m \cdot j} \le 2^{v}$.
\end{enumerate}
These two conditions follow by the requirements of the theorem. More specifically, the requirement that $a+j \cdot \ell \le \frac{2^{\ell}}{2}$ gives that $K \ge 2^{2^{\ell}/2} \ge v$. This gives that the first condition follows if $a+ j \cdot \ell \le j \cdot (m-r) -O(1)$ which is one of the requirements of the theorem.

The second condition follows because by our requirement $j \le \frac{2^{\ell}}{2}$ and $m \le 2^{\ell}$. Therefore for sufficiently large $\ell$, $m \cdot j \le 2^{2^{\ell}/2}=v$,
\end{proof}

\subsubsection{Revisiting the Nisan-Wigderson PRG and the Shaltiel-Umans PRG}
\label{sec:revisit}

As surveyed in Section \ref{sec:intro:bb proofs:hybrid} there are only two known constructions of black-box $\HFtoPRG$ proofs in the literature, where $\Red$ is useful and makes $q$ queries, for $q$ that can be significantly less than $2^{\ell}$. These are the Nisan-Wigderson PRG \cite{NW} and an additional construction is due to Shaltiel and Umans \cite{SU01} and Umans \cite{Umans02}.

The former is more versatile, and has many additional applications, mainly because its reduction requires less computational resources, and makes less queries.
The latter has advantages over the Nisan-Wigderson PRG as it allows to achieve $r=O(\ell)$ even for small values of $a$ (like $a=\poly(\ell)$, which corresponds to ``low-end'' PRGs).

Nevertheless, as explained in Section \ref{sec:intro:bb proofs:hybrid}, both these approaches rely on hardness amplification and the hybrid argument, preventing them from achieving the extreme high-end.

Theorem \ref{thm:fix} and \ref{thm:fix general} show that in \emph{any} black-box proof with $q \le 2^{\ell}$, \emph{even one that does not rely on the hybrid argument}, $\Fix_j(\Con)$ must be small. Let us review how the known PRGs achieve this.

\paragraph{The Nisan-Wigderson PRG.}
The Nisan-Wigderson PRG is a black-box $\rhoHFtoepsPRG{(\half+\frac{\eps}{m})}{\eps}$ proof $(\Con^{NW},\Red^{NW})$ with parameters $\ell,r,m,a,\rho=\half+\frac{\eps}{m}$ and $\eps$, such that $\Red$ makes a \emph{single} query. This means that the Nisan-Wigderson PRG is applicable in the extreme high-end, when starting from a $(\half+\frac{\eps}{m})\mbox{-}\HF$. As explained in Section \ref{sec:intro:results:hardness} current hardness amplification techniques are not applicable for $m \ge 2^{\ell/2}$. This means that they cannot be used in the extreme high-end, even if we are willing to allow large seed length $r$.

By Theorem \ref{thm:fix general}, we have that:
\[ \Fix_j(\Con^{NW}) \le a + j \cdot (\log m +O(1)), \]
for constant $\eps>0$, and many values of $m$ and $j$. Let us review how $\Con^{NW}$ achieves this.

The construction $\Con^{NW}$ is defined
using a ``design'', namely a collection $S_1,\ldots,S_m \subseteq [r]$ such that for every $i \in [m]$, $|S_i|=\ell$, and then for every $f \in \Func_{\ell,1}$, $\Con^{NW}_f:\B^r \to \B^{m}$ is defined by
\[ \Con^{NW}_f(x)=f(x|_{S_1}), \ldots, f(x|_{S_m}). \]
Another version that is often considered is a ``seed extending'' version $\overline{\Con}^{NW}_f(x)=x,\Con^{NW}_f(x)$ that outputs $m+r$ bits.
The analysis of \cite{NW} shows that the parameter $a$ is determined by the sizes of pairwise intersections of sets in the design. More specifically, using an improved analysis by Raz, Reingold and Vadhan \cite{RRV} (for a condition called ``weak design'') it follows that taking
\[ a=\sum_{i \in [m]} \sum_{i' \ne i} 2^{|S_{i'} \cap S_i|} \le m \cdot 2^{\max_{i \ne i'}|S_i \cap S_{i'}|}, \]
suffices to obtain a reduction $\Red^{NW}$ which makes a \emph{single} query.

The weak design property can be used to show that for every $1 \le j \le 2^{\ell}$, \[ \Fix_j(\overline{\Con}^{NW}) \le a+ j. \] Loosely speaking, this is because the NW proof shows that for every $i$, after fixing $a$ bits of information about $f$, we have that for every value of $x|_{[r] \setminus S_i}$, all $m-1$ output bits of the form $f(x|_{S_{i'}})$ for $i' \ne i$ are completely determined as a function of $x|_{S_i}$. This means that for every value of $x|_{S_i}$, the only output bit that is not yet fixed is $f(x|_{S_i})$, which can be fixed paying only
\emph{one bit of information}.

If we were to beat the hybrid argument, and construct a new black-box $\HFtoPRG$ proof that works directly from worst-case hard functions, we would have to come close to this behavior, and we hope that understanding this property may point us to new constructions, potentially utilizing the ability to ask a small number of queries (but more than just one query).

\paragraph{The Shaltiel-Umans PRG.}
The Shaltiel-Umans PRG is a black-box $\HFtoPRG$ proof $(\Con^{SU},\Red^{SU})$ with parameters $\ell,r,m,a,\rho=1$ and $\eps=\frac{1}{10}$, such that $\Red$ makes $q=m^{O(1)}$ queries.

Unlike the Nisan-Wigderson PRG, the Shaltiel-Umans PRG is stated for $\rho=1$ (starting from worst-case hard functions). Nevertheless, the Shaltiel-Umans PRG builds on the hardness amplification techniques of Sudan, Trevisan and Vadhan \cite{STV} and makes $q \ge m$ queries. Furthermore, it implies hardness amplification to $\half+\frac{1}{m}$. This means that it is not applicable at the extreme high-end.

By Theorem \ref{thm:fix general}, we have that:
\[ \Fix_j(\Con^{SU}) \le a + j \cdot (O(\log m)), \]
for many values of $m$ and $j$.

Let us review how $\Con^{SU}$ achieves this. The construction $\Con^{SU}_f$ relies on (a carefully chosen) black-box $\HFtorhoHF{(\half+\frac{1}{m})}$ proof (which has additional structure). This hardness amplification construction converts the worst-case hard function $f:\B^{\ell} \to \B$ into an average case hard function $f':\B^{\ell'} \to \B$.

The PRG construction is then defined by setting $r=\ell'$ and defining:
\[ \Con^{SU}_f(x)=f'(x),f'(g(x)),f'(g(g(x))),\ldots,f'(g^{(m-1)}(x), \]
where $g:\B^{\ell'} \to \B^{\ell'}$ is some specific function. Another version that is often considered is a ``seed extending'' version $\overline{\Con}^{SU}_f(x)=x,\Con^{SU}_f(x)$ that outputs $m+r$ bits.

This means that for every seed $x \in \B^{\ell'}$ if we consider the $\ell$ ``consecutive'' seeds, \[ g(x),g(g(x)),\ldots,g^{(m-1)}(x),\] all the outputs of $\Con^{SU}_f$ on these seeds, depend only on the values of $f'$ on $g(x),g(g(x)),\ldots,g^{(2m)}(x)$. This means that $m$ outputs of $\Con^{SU}_f$ can be fixed at the cost of fixing $2m$ bits of $f'$. In particular, for $j=a$ we get that $j=a$ outputs of $\Con_f$ can be fixed at the total cost of $2a=a + j \cdot 1$ bits of information about $f$. This means that after paying a fixed cost of $a$, $a$ outputs can be fixed at amortized cost of one bit of information per output. This can be used to show that:
\[ \Fix_a(\overline{\Con}^{SU}) \le a +a \cdot 1. \]
Again, if we were to beat the hybrid argument, and construct a new black-box $\HFtoPRG$ proof that works directly from worst-case hard functions, we may hope to take this understanding, hoping to somehow avoid paying the price of hardness amplification to extremely hard on average functions.

\subsection{Limitations on the ``PEG + Extractor'' approach of \cite{DMOZ}}
\label{sec:thm:PEG}

In this section we restate Theorem \ref{thm:PEG} and discuss its consequences and interpretation.

\subsubsection{Limitation on black-box $\HFtoPEG$ constructions with short seed}

\paragraph{Review of the setup of Theorem \ref{thm:PEG}.}
Theorem \ref{thm:fix} shows that for any black-box $\HFtoPEG$ proof $(\Con,\Red)$, if $\Red$ is useful, and makes $q \le 2^{\ell}$ queries, then $r \ge (1-o(1)) \cdot \ell \ge (1-o(1)) \cdot \log m$. As explained in Section \ref{sec:intro:results:PEG}, this shows that the ``PEG + extractor'' approach of \cite{DMOZ} cannot be used to obtain an extreme high-end PRG from \textsc{the extreme high-end hardness assumption} if the PEG is constructed by a black-box $\HFtoPEG$ proof.

\paragraph{Formal definition of black-box $\HFtoPEG$ proofs.} In the introduction, we did not give a formal definition of black-box $\HFtoPEG$ proofs (and only explained how these are defined in relation to black-box $\HFtoPRG$ proofs. We therefore start with a formal definition.

\begin{definition}[Black-box $\HFtoPEG$ proof]
\label{dfn:bb:hard-function=>PRG}
Given parameters $\ell,r,m,a,k,\eps,\rho$ a \remph{black-box $\rhoHFtokepsPEG{\rho}{k}{\eps}$ proof} is a pair $(\Con,\Red)$ of:
\begin{itemize}[noitemsep]
\item A construction map $\Con: \Func_{\ell,1} \to \Func_{r,m}$.
\item An oracle procedure $\Red^{(\cdot)}(x,\alpha)$ such that for every $f \in \Func_{\ell,1}$ and for every $D \in \Func_{m,1}$ such that $D$ $(k,\eps)$-$\PEGB$ $\Con_f$, there exists $\alpha \in \B^a$ such that the function $C \in \Func_{\ell,1}$ defined by $C(x)=\Red^D(x,\alpha)$, $\rho$-$\HFB$ $f$.
\end{itemize}
If we omit $\rho$, we mean $\rho=1$. If we omit $\eps$, we mean $\eps=1/10$.
We say that $\Red$ makes $q$ queries, if for every $D \in \Func_{m,1}$, $\alpha \in \B^a$, and $x \in \B^{\ell}$, $\Red^D(x,\alpha)$ makes at most $q$ oracle queries.
\end{definition}

\paragraph{A more general statement.}
The following Theorem is a generalized version of Theorem \ref{thm:PEG} which also allows the parameter $\rho$ (measuring how hard on average is the function we start from) to be very close to $\half$.

\begin{theorem}
\label{thm:PEG general}
There exists constants $\nu>0$ and $c>1$ such that for every
black-box $\rhoHFtokepsPEG{\rho}{k}{\eps}$ proof $(\Con,\Red)$ with parameters $\ell,r,m,a,k,\eps,\rho$ such that \[r < \ell - \log \ell - 2 \cdot \log \frac{1}{\eta} -c.\] If
$\rho=\half+\eta$, $\eta \ge 2^{-\ell}$, $k > r$, $\eps \le 1-2^{r-m}$, and $a \le \nu \cdot \eta^2 \cdot 2^{\ell}$,
then $\Red$ must make at least $q > 2^{\ell}$ queries.
\end{theorem}

\paragraph{A discussion of the parameters of Theorem \ref{thm:PEG general}.}
The parameters in Theorems \ref{thm:PEG} and Theorem \ref{thm:PEG general} are applicable even for very large $a$, and therefore Theorem \ref{thm:fix general} applies in the extreme high-end (where $a=2^{(1-o(1)) \cdot \ell}$) as well as in less challenging ranges such as the high-end (where $a=2^{\nu \cdot \ell}$ for a constant $\nu>0$) and the low-end (where $a=\poly(\ell)$). We stress once again that to the best of our knowledge, no known lower bound on the number of queries of a black-box reduction (of any kind) applies in the extreme high-end. See Section \ref{sec:hardness amplification} for a discussion on past lower bounds on the number of queries by reductions for black-box hardness amplification proofs.

We also note that the requirements on the parameters in Theorem \ref{thm:PEG general} are quite mild. We allow $\eps$ to approach one (rather than zero) and $k$ to approach $r$ (rather than $m$). For $\rho=1$ (namely, starting from a worst-case hard functions) we have that $\eta=\half$ and the theorem works even for $a=\Omega(2^{\ell})$.

\subsection{Proofs of Theorem \ref{thm:fix general} and Theorem \ref{thm:PEG general}}
\label{sec:prf:fix and PEG}

In this section we prove Theorem \ref{thm:fix general} and Theorem \ref{thm:PEG general} (which in turn imply Theorem \ref{thm:fix} and Theorem \ref{thm:PEG}).
Both Theorems will follow from the following theorem (which can be seen as a version of Theorem \ref{thm:fix general} for PEGs).

\begin{theorem}
\label{thm:PEG fix general}
There exists a constant $\nu>0$ such that for every
black-box $\rhoHFtokepsPEG{\rho}{k}{\eps}$ proof $(\Con,\Red)$ with parameters $\ell,r,m,a,k,\eps,\rho$ such that $\Red$ makes at most $q \le 2^{\ell}$ queries. If $\rho=\half+\eta$, $\eta \ge 2^{-\ell}$, $k > r$, $\eps \le 1-2^{r-m}$, and $a \le \nu \cdot \eta^2 \cdot 2^{\ell}$, then for $j_{\max} = \nu \cdot \frac{\eta^2 \cdot 2^{\ell}}{\ell}$, and every $j \le j_{\max}$, \[\Fix_j(\Con) \le a + j \cdot (\log q + \log \frac{4}{\eta}). \]
\end{theorem}

\noindent
The proof of Theorem \ref{thm:PEG fix general} is given in Section \ref{sec:prf:thm:fix general}.

\paragraph{Theorem \ref{thm:fix general} follows from Theorem \ref{thm:PEG fix general}.}
Loosely speaking, this follows because any PRG is also a PEG, and therefore limitations on PEGs also hold for PRGs. More formally, any black-box $\rhoHFtoepsPRG{\rho}{\eps}$ proof $(\Con,\Red)$ in which $\eps=1-2^{r-m}$ is also a black-box $\rhoHFtokepsPEG{\rho}{k}{\eps}$ proof for $k=r$. Theorem \ref{thm:PEG fix general} applies for the latter, and therefore also applies for the former, proving Theorem \ref{thm:fix general}.

\paragraph{Theorem \ref{thm:PEG general} follows from Theorem \ref{thm:PEG fix general}.}
The conditions in Theorem \ref{thm:PEG general} satisfy the conditions of Theorem \ref{thm:PEG fix general} except for the requirement made in Theorem \ref{thm:PEG fix general} that $q \le 2^{\ell}$. Therefore, under the conditions of Theorem \ref{thm:PEG fix general} we can apply Theorem \ref{thm:PEG fix general} and conclude that either $q > 2^{\ell}$ or the conclusion of Theorem \ref{thm:PEG fix general} holds and for $j_{\max} = \nu \cdot \frac{\eta^2 \cdot 2^{\ell}}{\ell}$, and every $j \le j_{\max}$, $\Fix_j(\Con) \le a + j \cdot (\log q + \log \frac{4}{\eta})$. However, the latter option cannot happen. This is because by the restriction on $r$, we get that $2^r < j_{\max}$. As the number of outputs of $\Con_f$ is $2^r < j_{\max}$, it is impossible to fix $j_{\max}$ distinct outputs, and so $\Fix_j(\Con)$ is undefined and can be viewed as infinity. Consequently, it can't be the case that $\Fix_j(\Con)$ is smaller than some finite quantity. We conclude that $q>2^{\ell}$.

\subsubsection{Proof of Theorem \ref{thm:PEG fix general}}
\label{sec:prf:thm:fix general}

In this section we prove Theorem \ref{thm:PEG fix general}. We assume the assumption of Theorem \ref{thm:PEG fix general}, namely that $(\Con,\Red)$ is a black-box $\rhoHFtokepsPEG{\rho}{k}{\eps}$ proof that satisfies the conditions of Theorem~\ref{thm:PEG fix general}.
We first define a simple distinguisher that answers one iff its input is a pseudorandom string.

\begin{definition}[The simple distinguisher]
For every $f \in \Func_{\ell,1}$ we define $D_f:\B^m \to \B$ by:
\[ D_f(z) =  \left\{\begin{array}{ll}
        1 , & \exists y \in \B^{r} \mbox{ s.t.\ } \Con_f(y)=z \\
        0, & \mbox{otherwise}
        \end{array}\right. \]
\end{definition}

\begin{claim}
\label{clm:exists A0}
There exists $\alpha' \in \B^a$ and $A_0 \subseteq \Func_{\ell,1}$ such that:
\begin{enumerate}
\item $\Pr_{F \from \Func_{\ell,1}}[F \in A_0] \ge 2^{-a}$.
\item For every $f \in A_0$, the function $C_f:\B^{\ell} \to \B$ defined by $C_f(x)=\Red^{D_f}(x,\alpha')$, satisfies that $C_f$ $\rho\mbox{-}\HFB$ the function $f$.
\end{enumerate}
\end{claim}

\begin{proof}
For every $f \in \Func_{\ell,1}$, $\Pr[D_f(\Con_f(U_r))=1]=1$, and $\Pr[D_f(U_m)=1] \le 2^{r-m}$. Therefore, $D_f$ $(k,\eps)-\PEGB$ $\Con_f$ if $k > r$ and $\eps \ge 1-2^{r-m}$, which is assumed in the statement of Theorem \ref{thm:PEG fix general}.
By the definition of black-box $\rhoHFtokepsPEG{\rho}{k}{\eps}$ proof, we have that for every $f \in \Func_{\ell,1}$, as $D_f$ $(k,\eps)$-$\PEGB$ $\Con_f$, there exists $\alpha \in \B^a$ such that the function $C \in \Func_{\ell,1}$ defined by $C(x)=\Red^D(x,\alpha)$, $\rho$-$\HFB$ $f$. Therefore, by averaging, there exists $\alpha' \in \B^a$ that works for a $2^{-a}$ fraction of the functions $f \in \Func_{\ell,1}$, and let $A_0$ denote the subset of all functions $f \in \Func_{\ell,1}$ for which $\alpha'$ works.
\end{proof}

We will now define a function $C:\B^{\ell} \ar \B$ (that does not depend on $f$) with the hope of showing that if the conclusion of Theorem \ref{thm:PEG fix general} does not hold, then $C(x)$ simulates $\Red^{D_f}(x,\alpha')$ quite well for ``some'' choices of $f \in \Func_{\ell,1}$ and $x \in \B^{\ell}$.

\begin{definition}
For every $x \in \B^{\ell}$, we define $z^x_1,\ldots,z^x_q \in \B^m$ as follows: We consider an invocation of $\Red^{(\cdot)}(x,\alpha')$ in which all queries are answered by zero, and for $1 \le i \le q$, let $z^x_i$ denote the $j$'th query made by $\Red(x,\alpha')$ in this invocation. Let $C(x)$ denote the output of $\Red(x,\alpha')$ in this invocation.
\end{definition}

Note that $z^x_1,\ldots,z^x_q \in \B^m$ do not depend on the choice of the function $D$ given to $\Red$ as oracle. We will be interested in the case that $\Red$ gets oracle access to the function $D_f$, for some $f \in \Func_{\ell,1}$. Note that for every $f \in \Func_{\ell,1}$, $z^x_1$ is the first query made by $\Red^{D_f}(x,\alpha')$, and for $i>1$, $z^x_i$ may be different than the $i$'th query made by $\Red^{D_f}(x,\alpha')$), as the reduction is allowed to make adaptive queries. Nevertheless, the following obviously holds:

\begin{claim}
\label{clm:adaptive}
For every $f \in \Func_{\ell,1}$ and every $x \in \B^{\ell}$, if $D_f(z^x_1)=\ldots=D_f(z^x_q)=0$ then $C(x)=\Red^{D_f}(x,\alpha')$.
\end{claim}

By definition, for every $f \in \Func_{\ell,1}$, we have that the number of $z \in \B^m$ on which $D_f$ answers one is small (at most $2^r$). We will say that a $z \in \B^m$ is \emph{weak} with respect to some $A \subseteq \Func_{\ell,1}$, if it is likely that $D_f(z)=1$ when $f$ is chosen uniformly in $A$.

\begin{definition}
We say that $z \in \B^m$ is $t$-weak with respect to a set $A \subseteq \Func_{\ell,1}$ if
\[ \Pr_{F \from A}[D_F(z)=1] \ge 2^{-t}. \]
\end{definition}

We will now consider an iterative process in which we will iteratively fix outputs $z$ of $\Con$ that are weak.

More precisely,
we set $t=\log \frac{4 \cdot q}{\eta}$ and consider the following iterative process. We start with the set $A_0$ that we already obtained, and $W_0=\emptyset$. We will maintain the following  invariant:

\begin{description}
\item[Invariant:] At step $j$ we maintain that:
\begin{enumerate}[noitemsep]
\item $\Pr_{F \from \Func_{\ell,1}}[F \in A_j] \ge 2^{-(a+jt)}$.
\item $A_j \subseteq A_0$.
\item $|W_j| = j$.
\item For every $z \in W_j$ and every $f \in A_j$, $D_f(z)=1$.
\end{enumerate}
\end{description}

Note that this invariant indeed holds for $j=0$.
At step $0 < j < j_{\max}$ we do the following:
If there does not exist a $z \not \in W_{j-1}$ that is $t$-weak with respect to $A_{j-1}$ then the process stops. Otherwise, if there exists $z \not \in W_{j-1}$ that is $t$-weak with respect to $A_{j-1}$, we define:
\begin{itemize}[noitemsep]
\item $W_j=W_{j-1} \cup \set{z}$.
\item $A_j=\set{f \in A_{j-1}: D_f(z)=1}$.
\end{itemize}
We observe that the invariant is indeed kept throughout this process.

\begin{claim}
\label{clm:invariant}
For every $j \le j_{max}$ for which the process has not yet stopped, the invariant above holds.
\end{claim}

\begin{proof}
This is obvious for the second, third and fourth items. The first item follows because for every $j$ if the process did not stop before step $j$, then $A_j$ exists, and we have that:
\begin{align*}
\Pr_{F \from \Func_{\ell,1}}[F \in A_j] &= \Pr_{F \from \Func_{\ell,1}}[F \in A_j|F \in A_{j-1}] \cdot \Pr_{F \from \Func_{\ell,1}}[F \in A_{j-1}] \\
& \ge \Pr_{F \from A_{j-1}}[F \in A_j] \cdot 2^{-(a+(j-1) \cdot t)} \\
& \ge 2^{-t} \cdot 2^{-(a+(j-1) \cdot t)} \\
& \ge  2^{-(a+j \cdot t)},
\end{align*}
where the third line is using $t$-weakness.
\end{proof}

\begin{claim}
\label{clm:exist A}
If the process stops at some $j^*<j_{max}$, then there exists $A \subseteq A_0$ such that:
\begin{enumerate}
\item $\Pr_{F \from \Func_{\ell,1}}[F \in A] \ge 2^{-(a+j^* \cdot t)+1}$.
\item $\Pr_{F \from A,X \in \B^{\ell}}[C(X)=F(X)] \ge \half + \frac{\eta}{2}$.
\end{enumerate}
\end{claim}

\begin{proof}
If the process stopped at some $j^*<j_{max}$, then there does not exist a $z \not \in W_{j^*}$ that is $t$-weak with respect to $A_{j^*}$. This in particular means that for every $x \in \B^{\ell}$, $z^x_1,\ldots,z^x_q$ are not $t$-weak with respect to $A_{j^*}$ and therefore, by a union bound:
\[ \Pr_{F \from A_{j^*}}[\exists i \in [q]: D_F(z^x_i)=1] \le q \cdot 2^{-t} \le \frac{\eta}{4}, \]
which means that:
\[ \Pr_{F \from A_{j^*}}[\forall i \in [q]: D_F(z^x_i)=0] \ge 1-\frac{\eta}{4}, \]
by Claim \ref{clm:adaptive} we conclude that for every $x \in \B^{\ell}$,
\[ \Pr_{F \from A_{j^*}}[C(x)=\Red^{D_F}(x,\alpha')] \ge 1-\frac{\eta}{4}. \]
For every $x \in \B^{\ell}$, let $V_x$ denote the random variable (over the probability space of $F \from A_{j^*}$) defined by:
\[ V_x =  \left\{\begin{array}{ll}
        1 , & C(x) \ne \Red^{D_F}(x,\alpha') \\
        0, & \mbox{otherwise}
        \end{array}\right. \]
Let $V = \sum_{x \in \B^{\ell}} V_x$. It follows that $\Exp_{F \from A_{j^*}}[V] \le 2^{\ell} \cdot \frac{\eta}{4}$. By Markov's inequality we have that:
\[ \Pr_{F \from A_{j^*}}[V> 2^{\ell} \cdot \frac{\eta}{2}] < \half. \]
Let $A=\set{f \in A_{j^*} : V  \le  2^{\ell} \cdot \frac{\eta}{2}}$. It follows that:
\[ \Pr_{F \from \Func_{\ell,1}}[F \in A] \ge \half \cdot \Pr_{F \from \Func_{\ell,1}}[F \in A_{j^*}] \ge 2^{-(a+j^* \cdot t)+1}. \]
For every $f \in A$, we have that the fraction of inputs $x \in \B^{\ell}$ on which $V_x(f)=1$ (meaning that $C(x) \ne \Red^{D_f}(x,\alpha')$) is $\frac{V(f)}{2^\ell} \le \frac{\eta}{2}$. This means that when we choose both $F \from A$ and $X \from \B^{\ell}$ independently, we have that:
\[ \Pr_{F \from A,X \from \B^{\ell}}[C(X) \ne \Red^{D_F}(X,\alpha')] \le \frac{\eta}{2}. \]
On the other hand, as $A \subseteq A_{j^*} \subseteq A_0$, by Claim \ref{clm:exists A0} we have that for every $f \in A$, the function $C_f:\B^{\ell} \to \B$ defined by $C_f(x)=\Red^{D_f}(x,\alpha')$, satisfies that $C_f$ $\rho$-$\HFB$ the function $f$. This means that for every $f \in A$,
\[ \Pr_{X \from \B^{\ell}}[\Red^{D_f}(X,\alpha')=f(X)] \ge \rho. \]
This means that when we choose both $F \from A$ and $X \from \B^{\ell}$ independently, we have that:
\[ \Pr_{F \from A,X \from \B^{\ell}}[\Red^{D_F}(X,\alpha')=F(X)] \ge \rho. \]
Putting things together, we have that:
\[ \Pr_{F \from A,X \from \B^{\ell}}[C(X)=F(X)] \ge \rho - \frac{\eta}{2}=\half+\frac{\eta}{2}. \]
\end{proof}

However, the next claim shows that if $A$ is a large set, it is unlikely that a single function $C$ is a good approximation to $F \from A$.

\begin{claim}
\label{clm:C fails}
For every $A \subseteq \Func_{\ell,1}$, and every function $C:\B^{\ell} \ar \B$, if
$\Pr_{F \from \Func_{\ell,1}}[F \in A] \ge 2^{-\Delta}$ then
\[ \Pr_{F \from A,X \from \B^{\ell}}[C(X)=F(X)] = \half+O\left(\sqrt{\frac{\Delta+\ell}{2^{\ell}}}\right) . \]
\end{claim}

\begin{proof}
For the purpose of contradiction we will set $\lambda = \Omega\left(\sqrt{\frac{\Delta+\ell}{2^{\ell}}}\right)$, and assume that
\[ \Pr_{F \from A,X \from \B^{\ell}}[C(X)=F(X)] \ge \half+\lambda. \]
By an averaging argument, it follows that:
\[ \Pr_{F \from A}\left[\Pr_{X \from \B^{\ell}}[C(X)=F(X)] \ge \half+\frac{\lambda}{2}\right] \ge \frac{\lambda}{2}. \]
Let $A' \subseteq A$ be the subset defined as follows:
\[A'=\set{f \in A: \Pr_{X \from \B^{\ell}}[C(X)=f(X)] \ge \half+\frac{\lambda}{2}}. \]
It follows that $|A'| \ge \frac{\lambda}{2} \cdot |A|$.
For every function $C \in \Func_{\ell,1}$, the number of $f \in \Func_{\ell,1}$ such that $\Pr_{X \from \B^{\ell}}[C(X)=f(X)] \ge \half+\frac{\lambda}{2}$ is bounded by the size of a Hamming ball in $\B^{2^{\ell}}$ that is of radius $(\half-\frac{\lambda}{2}) \cdot 2^{\ell}$. The latter quantity is bounded by $2^{H(\half-\lambda/2) \cdot 2^{\ell}} \le 2^{(1-O(\lambda^2)) \cdot 2^{\ell}}$ where $H(\cdot)$ is Shannon's binary entropy function, and using the fact that $H(\half-\lambda)=1-O(\lambda^2)$.
We conclude that there exists a constant $c>1$ such that:
\[ |A| \le \frac{2}{\lambda} \cdot |A'| \le \frac{2}{\lambda} \cdot 2^{(1-c \cdot \lambda^2) \cdot 2^{\ell}} = 2^{\log( \frac{1}{\lambda})+1 + (1-c \cdot \lambda^2) \cdot 2^{\ell}}= 2^{2^{\ell} - (c \cdot \lambda^2 \cdot 2^{\ell}-\log(1/\lambda)-1)}.  \]
We also have that $|A| \ge 2^{2^{\ell}-\Delta}$, and so we conclude that $\Delta > c \cdot \lambda^2 \cdot 2^\ell - \log(1/\lambda)-1$. This gives a contradiction if $\lambda =  \Omega\left(\sqrt{\frac{\Delta+\ell}{2^{\ell}}}\right)$.
\end{proof}

\noindent
Putting Claim \ref{clm:exist A} and Claim \ref{clm:C fails} together, we get that:

\begin{claim}
The iterative process does not stop until $j=j_{\max}$.
\end{claim}

\begin{proof}
By Claim \ref{clm:exist A} and Claim \ref{clm:C fails}, we conclude that if the process stops at some $j^*<j_{max}$ then there exists a constant $c>1$ such that:
\[\half + c \cdot \sqrt{\frac{a+j^* \cdot t + \ell}{2^{\ell}}} \ge \half+\frac{\eta}{2}, \]
Recalling that $t=\log \frac{4q}{\eta}$ this gives that:
\begin{align*}
q &\ge \frac{\eta}{4} \cdot 2^{\frac{\frac{2^{\ell} \cdot \eta^2}{4c}-a - \ell}{j^*}} \\
& \ge \frac{\eta}{4} \cdot 2^{\frac{\frac{2^{\ell} \cdot \eta^2}{5c}}{j^*}} \\
& > 2^{\ell},
\end{align*}
where the first inequality follows because we have that $a \le \nu \cdot \eta^2 \cdot 2^{\ell}$, and we can choose the constant $\nu>0$ to be sufficiently small. The second inequality follows because $\eta \ge 2^{-\ell}$, and $j^* \le j_{\max} \le \nu \cdot \frac{\eta^2 \cdot 2^{\ell}}{\ell}$, and we can choose $\nu>0$ to be sufficiently small.

Therefore, as we are assuming that $q \le 2^{\ell}$, it is a contradiction if the process stops at some $j^* < j_{\max}$.
\end{proof}

As the process did not stop at any $j < j_{\max}$, then for every $j \le j_{\max}$, the sets $A_j,W_j$ are defined, and by Claim \ref{clm:invariant} they maintain the invariant.

By the invariant, for every $j \le j_{\max}$ we conclude that for every $z \in W_j$ and every $f \in A_j$, $D_f(z)=1$ (meaning that there exists $y \in \B^r$ such that $z$ is an output of $\Con_f)$. By the invariant, we know that $|W_j|=j$ therefore, if we denote the elements of $W_j$ by $z_1,\ldots,z_j \in \B^m$ we conclude that
for every $f \in A_j$, and every $i \in [j]$ there exists $y_i \in \B^r$ (that may depend on $f$) such that $\Con_f(y_i)=z_i$. From the invariant, we also have that:
\[\Pr_{F \from \Func_{\ell,1}}[F \in A_j] \ge 2^{-(a+j \cdot t)}=2^{-(a+j \cdot (\log q + \log \frac{4}{\eta}))}. \]
Putting everything together we get that:
\[ \Pr_{F \from \Func_{\ell,1}}[\forall i \in [j]: \exists y_i \in \B^r \mbox{ s.t. }\Con_F(y_i)=z_i] \ge 2^{-(a+j \cdot (\log q + \log \frac{4}{\eta}))}, \]
which gives that for every $j \le j_{\max}$, $\Fix_j(\Con) \le a+j \cdot (\log q + \log \frac{4}{\eta})$. This proves Theorem \ref{thm:PEG fix general}.

\section{Limitations on black-box hardness amplification at the extreme high-end} \label{sec:hardness amplification}

In this section we prove Theorem \ref{thm:hardness} showing that reductions for black-box $\HFtorhoHF{(\half+\eps)}$ proofs must make many queries, even at the extreme high-end.
Theorem \ref{thm:hardness} is a combination of two lower bounds, stated next:

\begin{theorem}[Lower bound in terms of $\eps$]
\label{thm:hardness eps}
If $(\Con,\Red)$ is a $\HFtorhoHF{(\half+\eps)}$ proof
with parameters $\ell,\ell',a,\rho=1,\rho'=\half+\eps$, satisfying $a \le \frac{2^{\ell}}{10}$ and $\ell' \ge \log \frac{1}{\eps^2} + \Omega(1)$ then $\Red$ must make at least $q = \Omega(\frac{1}{\eps})$ queries.
\end{theorem}

\begin{theorem}[Lower bound in terms of $\ell$]
\label{thm:hardness ell}
If $(\Con,\Red)$ is a $\HFtorhoHF{(\half+\eps)}$ proof
with parameters $\ell,\ell',a,\rho=1,\rho'=\half+\eps$, satisfying $\eps \le \frac{1}{10}$, and $\ell' \ge \log \frac{1}{\eps^2} + \Omega(1)$ then $\Red$ must make at least $q \ge \frac{\ell-\log(2a)}{3}$ queries
\end{theorem}

A quantitatively better lower bound of $q=\Omega(\frac{\ell}{\eps^2})$ was proven by Grinberg, Shaltiel and Viola \cite{GSV18} for the case that $a \le 2^{\nu \cdot \ell}$ for some constant $\nu>0$. Theorems \ref{thm:hardness eps} and Theorem \ref{thm:hardness ell} achieve a smaller bound on $q$, but apply in the extreme high-end (where $a=2^{(1-o(1)) \cdot \ell}$) all the way up to $a = \frac{2^{\ell}}{10}$. This is especially significant in the case of Theorem \ref{thm:hardness eps} which (as we explained in detail in Section \ref{sec:intro:results:HF}) can be used to show limitations on the PRG composition of Chen and Tell \cite{CT1}.

\paragraph{Roadmap for this section.}
Both Theorem \ref{thm:hardness eps} and Theorem \ref{thm:hardness ell} will be proven by first connecting a black-box reduction to a depth 3 circuit for a version of the ``coin problem''. This connection is stated and proven in Section \ref{sec:reduction to circuit}.
The proofs of Theorem \ref{thm:hardness eps} and Theorem \ref{thm:hardness ell} show that such a depth 3 circuit must have large $q$. These proofs are given in Section \ref{sec:prf:thm:hardness eps} and Section \ref{sec:prf:thm:hardness ell}.

\subsection{Reductions as depth 3 circuits}
\label{sec:reduction to circuit}

The proof of Theorem \ref{thm:hardness eps} and Theorem \ref{thm:hardness ell} rely on the following lemma, which is inspired by an argument of Applebaum et al. \cite{AASY15} (see also \cite{RSV21}) and relates black-box reductions to constant depth circuits.
\begin{lemma}[Reduction to circuit]
\label{lem:Red => C}
Let $(\Con,\Red)$ be a black-box $\HFtorhoHF{(\half+\eps)}$ proof with parameters
$\ell,\ell',a,\rho=1,\rho'=\half+\eps$, such that $\Red$ makes $q$ queries.
For every $\alpha \in \B^a$, there exists a $q$-CNF $C^{\alpha}$ with $2^{q+\ell}$ clauses, over $n=2^{\ell'}$ variables, such that for $C=\bigvee_{\alpha \in \B^a} C^{\alpha}$, we have that:
\begin{itemize}[noitemsep]
\item $\Pr[C(U_n)] \le 2^{-(2^{\ell}-a)}$.
\item For every $x \in \B^n$ such that $\wt(x) \le (\half-\eps) \cdot n$, $C(x)=1$.
\end{itemize}
\end{lemma}

\begin{proof}
We will view a string $z$ of length $n=2^{\ell'}$, as a function $z \in \Func_{\ell',1}$ and vice-versa, by $z(y)=z_y$.
The proof of \cite{AASY15} (and the proof presented here) makes use of an idea originating in \cite{ViolaThesis,SV08} (and credited to Madhu Sudan) that the reduction must succeed if given oracle access to $\Con_f \oplus z$ for a string $z$ with $\wt(z) \le n \cdot (\half-\eps)$ (as such an oracle $(\half+\eps)$-$\HFB$ $\Con_f$), but cannot succeed when given oracle access to $\Con_f \oplus z$ for $z \from U_n$ (as $z$ ``wipes out'' the information in $\Con_f$). This will translate into the two conditions in the lemma. Details follow:

For every $f \in \Func_{\ell,1}$, $\alpha \in \B^a$ and $x \in \B^{\ell}$, we define $C_{f,\alpha,x}:\B^n \to \B$ as follows:
\[ C_{f,\alpha,x}(z)=1 \mbox{ iff $\Red^{\Con_f \oplus z}(x,\alpha)=f(x)$.} \]
We have that $\Red$ makes at most $q$ queries, and this implies that for every $f$ and $\alpha$, the function $C_{f,\alpha}$ can be computed by a decision tree with height $q$. Therefore, it can be computed by a $q$-CNF with $2^q$ clauses.
For every $f \in \Func_{\ell,1}$, and $\alpha \in \B^a$ we define $C_{f,\alpha}:\B^n \to \B$ as follows:
\[ C_{f,\alpha}(z)=1 \mbox{ iff $\forall x \in \B^{\ell}$, $\Red^{\Con_f \oplus z}(x,\alpha)=f(x)$.} \]
By definition for every $f$ and $\alpha$, $C_{f,\alpha}(x)=\bigwedge_{x \in \B^{\ell}}C_{f,\alpha,x}$. This means that $C_{f,\alpha}$ is an AND of $2^{\ell}$ $q$-CNFs with $2^q$ clauses, and overall, it can be computed by a $q$-CNF with $2^{q + \ell}$ clauses.
For every $f \in \Func_{\ell,1}$, we define $C_{f}:\B^n \to \B$ as follows:
\[ C_{f,\alpha}(z)=1 \mbox{ iff $\exists \alpha \in \B^a$, s.t. $\forall x \in \B^{\ell}$, $\Red^{\Con_f \oplus z}(x,\alpha)=f(x)$.} \]
This gives that for every $f \in \Func_{\ell,1}$, $C_f$ is an OR of $2^{a}$ $q$-CNFs with $2^{q+\ell}$ clauses.

Furthermore, by the definition of black-box $\HFtorhoHF{(\half+\eps)}$ proof, for every $f \in \Func_{\ell,1}$, if $\wt(z) \le (\half-\eps) \cdot n$, then the function $\Con_f \oplus z$, $\half+\eps$-$\HFB$ breaks $\Con_f$. This in turn implies (by Definition \ref{dfn:bb}) that there exists $\alpha \in \B^a$ such that for every $x \in \B^{\ell}$, $\Red^{\Con_f \oplus z}(x)=f(x)$, meaning that $C_f(z)=1$. This means that for every choice of $f \in \Func_{\ell,1}$, $C_f$ satisfies the second item.

We will now show that there exists an $f \in \Func_{\ell,1}$ such that $C_f$ satisfies the first item.
For a uniformly chosen $z \from U_n$, we have that for every $f \in \Func_{\ell,1}$, $\Con_f \oplus z$ is distributed uniformly over $\B^n$, and contains no information about $f$. It follows that:
\begin{align*}
\Pr_{f \from \Func_{\ell,1},z \from U_n}[C_f(z)=1] &= \Pr_{F \from \Func_{\ell,1},z \from U_n}[\mbox{$\exists \alpha \in \B^a$, s.t. $\forall x \in \B^{\ell}$, $\Red^{\Con_f \oplus z}(x,\alpha)=f(x)$}] \\
&= \Pr_{f \from \Func_{\ell,1},z \from U_n}[\mbox{$\exists \alpha \in \B^a$, s.t.  $\forall x \in \B^{\ell}$, $\Red^{z}(x,\alpha)=f(x)$} ]  \\
&\le Pr_{f \from \Func_{\ell,1}}[\mbox{$\exists \alpha \in \B^a$, s.t.  $\forall x \in \B^{\ell}$, $\Red^{z^*}(x,\alpha)=f(x)$} ],
\end{align*}
for some $z^* \in \B^n$ which maximizes the success probability. We can therefore continue and obtain that:
\begin{align*}
\Pr_{f \from \Func_{\ell,1},z \from U_n}[C_f(z)=1] &\le \Pr_{F \from \Func_{\ell,1}}[\mbox{$\exists \alpha \in \B^a$, s.t.  $\forall x \in \B^{\ell}$, $\Red^{z^*}(x,\alpha)=f(x)$}, ] \\
&\le \sum_{\alpha \in \B^a}\Pr_{f \from \Func_{\ell,1}}[\mbox{$\forall x \in \B^{\ell}$, $\Red^{z^*}(x,\alpha)=f(x)$} ]  \\
&\le 2^a \cdot \frac{1}{2^{2^{\ell}}},
\end{align*}
where the penultimate inequality follows by a union bound, and the last inequality follows because for every $x \in \B^{\ell}$,
\[ \Pr_{f \from \Func_{\ell,1}}[\Red^{z^*}(x,\alpha)=f(x) ] = \half, \]
and these events are independent for the $2^{\ell}$ choices of $x \in \B^{\ell}$.
Finally, by averaging, we conclude that there exists $f \in \Func_{\ell,1}$ such that:
\[ \Pr_{z \from U_n}[C_f(z)=1] \le 2^{-(2^{\ell}-a)}. \]
The final function $C$ will be this function $C_f$ which indeed satisfies the properties in the conclusion of the lemma.
\end{proof}

\subsection{Proof of Theorem \ref{thm:hardness eps}}
\label{sec:prf:thm:hardness eps}

In this section we prove Theorem \ref{thm:hardness eps}. We split the proof into two parts, specified in Lemma \ref{lem:starting} and Lemma \ref{lem:continuing} below. This splitting will allow us to use Lemma \ref{lem:continuing} in the application to lower bounds on local list-decoding in Section \ref{sec:codes}.

\begin{lemma}
\label{lem:starting}
Assume the conditions of Theorem \ref{thm:hardness eps}. For every $\alpha \in \B^a$, there exists a $q$-CNF $C^{\alpha}$ over $n=2^{\ell'}$ variables, such that for $C=\bigvee_{\alpha \in \B^a} C^{\alpha}$, and $p=\frac{\eps}{10}$:
\begin{itemize}[noitemsep]
\item $\Pr[C(U_n)=1] \le 2^{-(2^{\ell}-a)}$.
\item $\Pr_{\rho \from \Rest^n_p}[\Pr_{x \from U^{1/3}_{n \cdot p}}[C(\Fill_{\rho}(x))=1] \ge 0.99] \ge 0.99$.
\end{itemize}
\end{lemma}

The next lemma shows that under (a weak form) of the conclusion of Lemma \ref{lem:starting}, if $a$ is slightly smaller than $2^{\ell}$ then $q=\Omega(\frac{1}{\eps})$.

\begin{lemma}
\label{lem:continuing}
Let $C$ be a circuit such that $C=\bigvee_{\alpha \in \B^a} C^{\alpha}$, where for each $\alpha \in \B^a$, $C^{\alpha}$ is a $q$-CNF over $n$ variables. Let $p=\frac{\eps}{10}$ and assume that:
\begin{itemize}[noitemsep]
\item $\Pr[C(U_n)=1] \le 2^{-(2^{\ell}-a)}$.
\item $\Pr_{\rho \from \Rest^n_p}\left[\Pr_{x \from U^{1/3}_{n \cdot p}}[C(\Fill_{\rho}(x))=1] \ge 0.01\right] \ge 0.01$.
\item $a \le \frac{2^{\ell}}{10}$.
\end{itemize}
Then
$q \ge \frac{1000}{\eps}$.
\end{lemma}

Together, Lemma \ref{lem:starting} and Lemma \ref{lem:continuing} imply Theorem \ref{thm:hardness eps}. The proofs of Lemma \ref{lem:starting} and Lemma \ref{lem:continuing} are given in Sections \ref{sec:prf:lem:starting} and \ref{sec:prf:lem:continuing}.

\subsubsection{Proof of Lemma \ref{lem:starting}}
\label{sec:prf:lem:starting}

We first apply Lemma \ref{lem:Red => C} and obtain that for every $\alpha \in \B^a$, there exists a $q$-CNF $C^{\alpha}$ over $n=2^{\ell'}$ variables, such that:
\begin{itemize}[noitemsep]
\item $\Pr[C^{\alpha}(U_n)=1] \le 2^{-(2^{\ell}-a)}$.
\item For every $x \in \B^n$ such that $\wt(x) \le (\half-\eps) \cdot n$, there exists $\alpha \in \B^a$ such that $C^{\alpha}(x)=1$.
\end{itemize}

When choosing $\rho \from \Rest^n_p$ we expect that $\frac{1}{n} \cdot |\set{i:\rho(i)=1}|=\half-\frac{p}{2}$. Therefore, by a Chernoff bound:
\[ \Pr_{\rho \from \Rest^n_p}[ \frac{1}{n} \cdot|\set{i:\rho(i)=1}| \ge \half-0.49 \cdot p] \cdot  \le 2^{-\Omega(p^2 \cdot n)},\]
which is smaller than $0.01$, by the assumption that $\ell' \ge \log \frac{1}{\eps^2} + \Omega(1)$, and $n=2^{\ell'}$.
If this event occurs, then for every $x \in \B^{n \cdot p}$,
\[\wt(\Fill_{\rho}(x)) \le (\half-0.49 \cdot p) \cdot n + \wt(x). \]
This means that if $\wt(x) \le 0.34 \cdot p \cdot n$, then
\[ \wt(\Fill_{\rho}(x)) \le n \cdot (\half- (0.49-0.34) \cdot p) \le n \cdot(\half- \eps), \]
By our assumption that $p \le \frac{\eps}{10}$. Applying a Chernoff bound, we conclude that:
\[ \Pr_{x \from U^{1/3}_{n \cdot p}}[\wt(x) \ge 0.34 \cdot p \cdot n] \le 2^{-\Omega(pn)}, \]
which is smaller than $0.01$ by the assumption that $\ell' \ge \log \frac{1}{\eps^2} + \Omega(1)$, and $n=2^{\ell'}$.

\subsubsection{Proof of Lemma \ref{lem:continuing}}
\label{sec:prf:lem:continuing}

We want to show that following a random restriction, each $C^{\alpha}$ is a low height decision tree.

\begin{claim}
If $q < \frac{1000}{\eps}$ then there exists a restriction $\rho:[n] \to \set{0,1,*}$ with $\Free(\rho)=n \cdot p$ such that:
\begin{itemize}[noitemsep]
\item For every $\alpha \in \B^a$, $C^{\alpha}_{\rho}$ is a decision tree of height $a$.
\item $\Pr[C_{\rho}(U_n) = 1] \le 2^{-(2^{\ell}-a-10)}$.
\item $\Pr[C_{\rho}(U^{1/3}_n)=1] \ge 0.01$.
\end{itemize}
\end{claim}

\begin{proof}
We will show that when choosing $\rho \from \Rest^n_p$ there is a positive probability to obtain a $\rho$ that satisfies all three items.
This will follow by a union bound in which we show that the probability that each of the items does not hold is small.

For the first item, we note that by Theorem \ref{thm:switching} for every $\alpha \in \B^a$, the probability over $\rho \from \Rest^n_p$ that $C^{\alpha}_{\rho}$ does not have a decision tree of height $a$, is at most $(7 \cdot p \cdot q)^a \le (7/100)^a$.
By a union bound over the $2^a$ choices of $\alpha \in \B^a$, the probability over $\rho \from \Rest^n_p$ that there exists $\alpha \in \B^a$ such that $C^{\alpha}_{\rho}$ does not have a decision tree of height $a$, is at most $0.001$.

For the second item, we recall that by Lemma \ref{lem:starting} we have that:
$\Pr[C(U_n)=1] \le 2^{-(2^{\ell}-a)}$.
Therefore, by Markov's in equality we conclude that with probability at least $0.999$ over $\rho \from \Rest^n_p$:
\[ \Pr[C_{\rho}(U_n) = 1] \le 1000\cdot 2^{-(2^{\ell}-a)} \le 2^{-(2^{\ell}-a-10)}.\]

\noindent
Finally, by the second item in Lemma \ref{lem:starting} we have that with probability at least $0.01$ over $\rho \from \Rest^n_p$:
\[ \Pr[C_{\rho}(U^{1/3}_n)=1] \ge 0.01 \]
Overall, by a union bound with probability at least $1-(0.99+0.001 + 0.001) >0$ over $\rho \from \Rest^n_p$, we obtain a $\rho$ that satisfies all three conditions.
\end{proof}

Let $\rho$ be the restriction guaranteed by the claim. We can view $C_{\rho}$ as $C_{\rho}=\bigvee_{\alpha \in \B^a} C^{\alpha}_{\rho}$. This means that $C_{\rho}$ is an OR of $2^a$ decision trees of height $a$. Each such decision tree can be replaced by an $a$-DNF with $2^a$ clauses, and overall, we get that $C_{\rho}$ is an $a$-DNF with $2^{2a}$ clauses. This gives that:

\begin{claim}
\label{clm:DNF distinguishes}
If $q < \frac{1000}{\eps}$ then there exists an $a$-DNF $C'$ with $p \cdot n$ variables and $2^{2a}$ clauses such that:
\begin{itemize}[noitemsep]
\item $\Pr[C'(U_{p \cdot n}) = 1] \le 2^{-(2^{\ell}-a-10)}$.
\item $\Pr[C'(U^{1/3}_{p \cdot n})=1] \ge 0.01$.
\end{itemize}
\end{claim}

This means that one of the $2^{2a}$ clauses of $C'$ must satisfy the following:

\begin{claim}
\label{clm:clause distinguishes}
If $q <\frac{1000}{\eps}$ then there exists a DNF clause $\bar{C}$ over $t \le a$ literals such that:
\begin{itemize}[noitemsep]
\item $\Pr[\bar{C}(U_t) = 1] \le 2^{-(2^{\ell}-a-10)}$.
\item $\Pr[\bar{C}(U^{1/3}_t)=1] \ge \frac{0.01}{2^{2a}} \ge 2^{-(2a+10)}$.
\end{itemize}
\end{claim}

We are now ready to prove Lemma \ref{lem:continuing}. We assume for the purpose of contradiction that $q < \frac{1000}{\eps}$ and show that Claim \ref{clm:clause distinguishes} yields a contradiction. This is because a DNF clause over $t \le a$ literals cannot distinguish $U_t$ from $U^{1/3}_t$ with the parameters stated in Claim \ref{clm:clause distinguishes}.

More precisely, as under $U^{1/3}_t$, every literal evaluates to one with probability that is upper bounded by $2/3$, we conclude that:
\[ \Pr[\bar{C}(U^{1/3}_t)=1] \le \left(\frac{2}{3}\right)^t. \]
Under $U_t$, every literal evaluates to one with probability half. Therefore,
\[ \Pr[\bar{C}(U_t)=1] = \left(\frac{1}{2}\right)^t. \]
By Claim \ref{clm:clause distinguishes} it follows that:
\[ \left(\frac{4}{3}\right)^t = \frac{\left(\frac{2}{3}\right)^t}{\left(\frac{1}{2}\right)^t} \ge \frac{2^{-(2a+10)}}{2^{-(2^{\ell}-a-10)}} = 2^{2^{\ell}-3a-20}. \]
which gives that:
\[ a \ge t \ge \frac{2^{\ell}-3a-20}{\log \frac{4}{3}} > \frac{2^{\ell}}{10}, \]
using the assumption that $a \le \frac{2^{\ell}}{10}$, and the lemma follows.

\subsection{Proof of Theorem \ref{thm:hardness ell}}
\label{sec:prf:thm:hardness ell}

In this section we prove Theorem \ref{thm:hardness ell}.
We first apply Lemma \ref{lem:Red => C} and obtain that for every $\alpha \in \B^a$, there exists a $q$-CNF $C^{\alpha}$ with $2^{q+\ell}$ clauses, over $n=2^{\ell'}$ variables, such that for $C=\bigvee_{\alpha \in \B^a} C^{\alpha}$, we have that:
\begin{itemize}[noitemsep]
\item $\Pr[C(U_n)=1] \le 2^{-(2^{\ell}-a)}$.
\item For every $x \in \B^n$ such that $\wt(x) \le (\half-\eps) \cdot n$, there exists $\alpha \in \B^a$ such that $C^{\alpha}(x)=1$.
\end{itemize}

\noindent
As there are $2^a$ such CNFs, it follows that one of them satisfies:

\begin{claim}
\label{clm:alpha distinguishes}
There exists $\alpha \in \B^a$ such that:
\begin{itemize}[noitemsep]
\item $\Pr[C^{\alpha}(U_n)=1] \le 2^{-(2^{\ell}-a)}$.
\item $\Pr[C^{\alpha}(U^{1/3}_n)=1] \ge 2^{-(a+1)}$.
\end{itemize}
\end{claim}

\begin{proof}
The first item follows because the top gate of $C$ is an OR gate.
For the second item, by a Chernoff bound, and our assumption that $\eps \le \frac{1}{10}$:
\[ \Pr_{x \from U^{1/3}_n}[\wt(x) \le (\half-\eps) \cdot n] \le 2^{-\Omega(n)} \le \half,\]
meaning that:
\[ \Pr_{x \from U^{1/3}_n}[C(x)=1] \ge \half,\]
and by averaging, there exists $\alpha \in \B^a$ such that:
\[ \Pr_{x \from U^{1/3}_n}[C^{\alpha}(x)=1] \ge \frac{1}{2 \cdot 2^a}. \]

\end{proof}

Let $C_0=C^{\alpha}$ that is guaranteed by Claim \ref{clm:alpha distinguishes}.
We will start from $C_0$ and will iteratively apply the following lemma (which is proven in Section \ref{sec:prf:lem:reduce}).

\begin{lemma}
\label{lem:reduce}
For every $v$-CNF $D$ over $m$ variables, if $\Pr[D(U^{1/3}_m)=1] \ge 2^{-w}$ then
there exists a restriction $\rho:[m] \to \set{0,1,*}$ with $\Free(\rho)=m'=m-w \cdot 4^v$, such that:
\begin{itemize}[noitemsep]
\item $D_{\rho}$ is a $(v-1)$-CNF over $m'$ variables.
\item $\Pr[D_\rho(U_{m'})=1] \le \Pr[D(U_m)=1] \cdot 2^{w \cdot 4^{v}}$.
\item $\Pr[D_\rho(U^{1/3}_{m'})=1] \ge \Pr[D(U^{1/3}_{m})=1]$.
\end{itemize}
\end{lemma}

More precisely, at every step $i>0$, we apply Lemma \ref{lem:reduce} with $w=a+1$ on $D=C_{i-1}$, and set $C_i=D_{\rho}$. We use $m_i$ to denote the input length of $C_i$. We will maintain the following invariant (and note that this indeed holds for $i=0$).

\begin{description}[noitemsep]
\item[Invariant:] At step $i$ the following hold:
\begin{itemize}[noitemsep]
\item $C_i$ is a $(q-i)$-CNF.
\item $\Pr[C_i(U_{m_i})=1] \le 2^{-\left(2^{\ell} - a-i \cdot (a+1) \cdot 4^{q}\right)}$.
\item $\Pr[C_i(U^{1/3}_{m_i})=1] \ge 2^{-(a+1)}$.
\end{itemize}
\end{description}

It immediately follows that this invariant holds for every $i \le q$.
This means that at the conclusion of this process, for $i=q$ we have that:
\begin{itemize}[noitemsep]
\item $C_m$ is a constant.
\item $\Pr[C_q(U_{m_q})=1] \le 2^{-\left(2^{\ell} - a- q \cdot (a+1) \cdot 4^{q}\right)}$.
\item $\Pr[C_q(U^{1/3}_{m_q})=1] \ge 2^{-(a+1)}$.
\end{itemize}

\noindent
Therefore, by the third item, it must be that $C_m$ is the constant one, and by the second item it must be that:
\[ 2^{-\left(2^{\ell} -a - q \cdot (a+1) \cdot 4^{q}\right)} \ge 1.\]
which implies that:
\[ q \ge \frac{\ell-\log(a+1)}{3}. \]

\subsubsection{Proof of Lemma \ref{lem:reduce}}
\label{sec:prf:lem:reduce}

In this section, we prove Lemma \ref{lem:reduce}. We first observe that if a $v$-CNF does not have many clauses that have disjoint variables, then there must be a small ``cover'' (namely, a set of variables such that each clause contains a variable from the cover).

\begin{proposition}
\label{prop:independence or cover}
In every $v$-CNF $D$, if there does not exist $t$ clauses with disjoint variables, then there is a subset $S$ of $s=(t-1) \cdot v$ variables, such that every clause of $D$ contains a variable in $S$.
\end{proposition}

\begin{proof}
We go over the clauses one by one, staring with an empty $S$. Every time the current clause does not contain a variable from $S$, we mark the current clause, and add its variables to $S$. Every time we mark a clause, its variables are disjoint from all variables of previously marked clauses. By our assumption, we mark at most $t-1$ clauses, and so when we conclude the size of $S$ is at most $(t-1) \cdot v$.
\end{proof}

\begin{lemma}
\label{lem:find cover}
For every $v$-CNF $D$ over $m$ variables, if $\Pr[D(U^{1/3}_m)=1] \ge 2^{-w}$ then
there exists a set $S$ of size $s=w \cdot 4^{v}$ variables such that every clause of $D$ contains at least one variable in $S$.
\end{lemma}

\begin{proof}
If there does not exists a set $S$ with the required properties, then by Proposition \ref{prop:independence or cover}, $D$ has $t=\frac{s}{v}$ clauses that have disjoint variables. Under $U^{1/3}_m$, each literal of $D$ has probability at least $1/3$ to evaluate to zero. This gives that each clause has probability at most $1-\left(\frac{1}{3}\right)^v$ to evaluate to one. Consequently, the probability that $t$ disjoint clauses evaluate to one is at most $\left(1-\left(\frac{1}{3}\right)^v\right)^t$, implying that:
\[ \Pr[D(U^{1/3}_m)=1] \le \left(1-\left(\frac{1}{3}\right)^v\right)^t \le e^{-3^{-v} \cdot t} < e^{-4^{-v}\cdot s} \le 2^{-w}. \]
and we get a contradiction.
\end{proof}

We are ready to prove Lemma \ref{lem:reduce}. By Lemma \ref{lem:find cover} there exists a set $S$ of size $s=w \cdot 4^{v}$ variables such that every clause of $D$ contains at least one variable in $S$. Let $A$ be the set of all restrictions that fix the variables in $S$, while leaving the other variables free. More precisely,
\[ A=\set{\rho:[m] \to \set{0,1,*}: \rho(i) \ne * \mbox{ iff } i \in S}. \]
We will show that there exists a $\rho \in A$ that satisfies the guarantees of Lemma \ref{lem:reduce} with positive probability.

The first item of Lemma \ref{lem:reduce} holds for every $\rho \in A$. This is because every clause of $D$ has a variable in $S$ and following the restriction, each clause is over at most $v-1$ variables.

The second item of Lemma \ref{lem:reduce} also holds for every $\rho \in A$. For every choice $y \in \B^{w \cdot 4^v}$ of the restricted bits in $\rho$, let $E_{\rho,y}$ denote the event  $\set{X|_{[n] \setminus \Select(\rho)}=y}$. For every $\rho \in A$ we have that:
\begin{align*}
\Pr[D_{\rho}[D(U_{m'})=1] &= \Pr_{X \from U_m}[D(X)=1|E_{\rho,y}] \\
&\le \frac{\Pr_{X \from U_m}[D(X)=1]}{\Pr_{X \from U_m}[
E_{\rho,y}]} \\
&= \Pr[D(U_m)=1] \cdot 2^{w \cdot 4^v}.
\end{align*}
The third item follows because by averaging, there exists $\rho \in A$ and $y \in \B^{w \cdot 4^v}$ such that $\Assign(\rho)=y$ and:
\[ \Pr_{X \from U^{1/3}_{m}}[D(X)=1|E_{\rho,y}] \ge \Pr_{X \from U^{1/3}_{m}}[D(X)=1], \]
For this choice of $\rho$ and $y$, we have that
\begin{align*}
\Pr[D_\rho(U^{1/3}_{m'})=1] &= \Pr_{X \from U^{1/3}_{m}}[D(X)=1|X \in E_{\rho,y}] \\
&\ge \Pr_{X \from U^{1/3}_{m}}[D(X)=1] \\
&= \Pr[D(U^{1/3}_{m'})=1].
\end{align*}
This concludes the proof of Lemma \ref{lem:reduce}.

\section{Improved lower bounds for local list-decoding algorithms}
\label{sec:codes}

In this section we show that the techniques developed in Section \ref{sec:hardness amplification} can be used to get an improved lower bounds on the number of queries of local list-decoding algorithms. In Section \ref{sec:codes:dfn} we review the definition of locally list-decodable codes, and local decoding algorithms. In Section \ref{sec:codes:results} we state our improved bound and compare it to previous work. In Section \ref{sec:codes:proof} we prove the improved bound.

\subsection{Definition of locally list-decodable codes}
\label{sec:codes:dfn}

List-decodable codes are a natural extension of (uniquely decodable) error-correcting codes, as it allows (list) decoding for error regimes where unique decoding is impossible. This is an extensively studied area; see \cite{Gur-survey06} for a survey. In this paper, we will be interested in list-decoding of binary codes.

\begin{definition}[List-decodable code]
For a function $\Enc:\B^k \ar \B^n$, and $w \in \B^n$, we define
\[ \List^{\Enc}_{\alpha}(w)=\set{m \in \B^k: \mathsf{dist}(\Enc(m),w) \le \alpha}. \]
We say that $\Enc$ is $(\alpha,L)$-list-decodable if for every $w \in \B^n$, $|\List^{\Enc}_{\alpha}(w)| \le L$.
\end{definition}
The task of \emph{algorithmic} list-decoding is to produce the list $\List^{\Enc}_{\alpha}(w)$ on input $w \in \B^n$.

\emph{Local} unique decoding algorithms are algorithms that given an index $i \in [k]$, make few oracle queries to $w$, and reproduce the bit $m_i$ (with high probability over the choice of their random coins), where $m$ denotes the unique codeword close to $w$. This notion of \emph{local decoding} has many connections and applications in computer science and mathematics~\cite{Yekhanin12}.

We will be interested in \emph{local} list-decoding algorithms
that receive oracle access to a received word $w \in \B^n$, as well as inputs $i \in [k]$ and $j \in [L]$.
We will require that for every $m \in \List^{\Enc}_{\alpha}(w)$, with high probability, there exists a $j \in [L]$ such that for every $i \in [k]$, when $\Dec$ receives oracle access to $w$ and inputs $i,j$, it produces $m_i$ with high probability over its choice of random coins. This motivates the next definition.

\begin{definition}[Randomized local computation]
\label{dfn:randomized local computation}
We say that a procedure $P(i,r)$ \remph{locally computes} a string $m \in \B^k$ with error $\delta$, if for every $i \in [k]$, $\Pr[P(i,R)=m_i] \ge 1-\delta$ (where the probability is over a uniform choice of the ``string of random coins'' $R$).
\end{definition}

The definition of local list-decoders considers an algorithmic scenario that works in two steps:
\begin{itemize}
\item At the first step (which can be thought of as a preprocessing step) the local list-decoder $\Dec$ is given oracle access to $w$ and an index $j \in [L]$. It tosses random coins (which we denote by $r^{\shared}$).
\item At the second step, the decoder receives the additional index $i \in [k]$, and tosses additional coins $r$.
\item It is required that for every $w \in \B^n$ and $m \in \List^{\Enc}_{\alpha}(w)$, with probability $2/3$ over the choice of the shared coins $r^{\shared}$, there exists $j \in [L]$ such that when the local list-decoder receives $j$, it locally computes $m$ (using its ``non-shared'' coins $r$). The definition uses two types of random coins because the coins $r^{\shared}$ are ``shared'' between different choices of $i \in [k]$ and allow different $i$'s to ``coordinate''. The coins $r$, are chosen independently for different choices of $i \in [k]$.
\end{itemize}

This is formally stated in the next definition.

\begin{definition}[Local list-decoder]
\label{dfn:LLD}
Let $\Enc:\B^k \ar \B^n$ be a function. An \remph{$(\alpha,L,q,\delta)$-local list-decoder (LLD)} for $\Enc$ is an oracle procedure $\Dec^{(\cdot)}$ that receives oracle access to a word $w \in \B^n$, and makes at most $q$ calls to the oracle. The procedure $\Dec$ also receives inputs:
\begin{itemize}
\item $i \in [k]$ : The index of the symbol that it needs to decode.
\item $j \in [L]$ : An index to the list.
\item Two strings $r^{\shared},r$ that are used as random coins.
\end{itemize}
It is required that for every $w \in \B^n$, and for every $m \in \List^{\Enc}_{\alpha}(w)$, with probability at least $2/3$ over choosing a uniform string $r^{\shared}$, there exists $j \in [L]$ such that the procedure \[P_{w,j,r^{\shared}}(i,r)=\Dec^w(i,j,r^{\shared},r) \] locally computes $m$ with error $\delta$. If we omit $\delta$, then we mean $\delta=1/3$.
\end{definition}

\noindent
See \cite{RSV21} for a discussion on the generality of this definition, and on past work in this area.

\subsection{Our Results}
\label{sec:codes:results}

Ron-Zewi, Shaltiel and Varma \cite{RSV21} showed lower bounds on the number of queries of local list decoders.
We use our improved techniques to prove the following theorem.

\begin{theorem}[Improved lower bounds for small $\eps$]
\label{thm:main:code:small eps}
There exist constants $c_1,c_2>1$ such that for every $L \le 2^{\frac{k}{20}}$, $\delta <\frac{1}{3}$, $n \ge \frac{c_2}{\eps^2}$, and every $(\half-\eps,L,q,\delta)$-local list-decoder for $\Enc:\B^k \ar \B^n$, we have that $q \ge \frac{1}{c_1 \cdot \log (k) \cdot \eps}$.
\end{theorem}

\noindent
The previous bounds of \cite{RSV21} come in two forms:
\begin{itemize}[noitemsep]
\item For $L \le 2^{k^{0.9}}$ and $\eps \ge \frac{1}{k^{\Omega(1)}}$, \cite{RSV21} obtain tight bounds (up to constants), showing that: $q = \Omega(\frac{\log(1/\delta)}{\eps^2})$.
\item If these conditions are not met (and in particular, if $\eps$ is smaller than $1/k$), \cite{RSV21} obtain weaker bounds, showing that if $\delta < \frac{1}{3}$, $q=\Omega(\frac{1}{\sqrt{\eps} \cdot \log k}) - O(\log L)$.
\end{itemize}

Theorem \ref{thm:main:code:small eps} improves upon the second item above (although it does not match the optimal bound of the first item). More specifically, Theorem \ref{thm:main:code:small eps} replaces $\sqrt{\eps}$ with $\eps$, and does not have the additive term of ``$-O(\log L)$''.

Not surprisingly, these improvements directly correspond to ``the extreme high-end''. Making this analogy, (namely, setting $L=2^a$, and $k=2^{\ell}$) we have that for $L=2^{k^{1-o(1)}}$, previous work does not give any bound if $\eps \ge \frac{1}{k^{1.9}}$, and in particular for $\eps \approx \frac{1}{k}$.
In contrast, Theorem \ref{thm:main:code:small eps} gives a bound that is $\Omega(\frac{1}{\log k \cdot \eps})$, which is not far from the known upper bound of $O(\frac{1}{\eps^2})$, and is polynomially related to the upper bound for $\eps \le \frac{1}{\log^2 k}$.

\subsection{Proof of Theorem \ref{thm:main:code:small eps}}
\label{sec:codes:proof}

In this section we prove Theorem \ref{thm:main:code:small eps}. The proof uses a similar approach as that of \cite{RSV21} (which is in turn based on \cite{AASY15}) to transform an LLD into a depth 3 circuit for a certain problem. (We make stronger requirements on the circuit, and therefore, need to redo the reduction, taking care to obtain these stronger requirements). Once we obtain a depth 3 circuit, we use Lemma \ref{lem:continuing} to show improved lower bounds. Details follow:

We assume that $L \le 2^{\frac{k}{20}}$, and $n \ge \frac{c_2}{\eps^2}$.
Our goal is to prove lower bounds on the number of queries $q$ of $(\half-\eps,L,q,\delta)$-local list-decoders for $\delta<\frac{1}{3}$. It is possible to amplify the error probability $\delta$ from $1/3$ to $1/20k$ as follows: After choosing the random string $r^{\shared}$, we choose $e=O(\log k)$ independent uniform strings $r_1,\ldots,r_e$, and apply $\Dec^{(\cdot)}(i,j,r_{\ell},r^{\shared})$ for all choices of $\ell \in [e]$. We then output the majority vote of the individual $e$ outputs. It is standard that this gives a $(\half - \eps, q'=O(q \cdot \log k), L, 1/20k)$-LLD for $\Enc: \B^k \to \B^n$.

We plan to reduce to Lemma \ref{lem:continuing}, and set $p=\frac{\eps}{10}$ (as in that lemma).
Our next step is to fix the random coins of the decoder, and obtain a \emph{deterministic} decoder $\overline{Dec}$ which succeeds in a specific experiment in which: a message $m \from \B^k$ is chosen uniformly, and a ``noise string'' $z$ is chosen by choosing $\rho \from \Rest^n_p$ and filling the unrestricted variables with $U^{1/3}_{p \cdot n}$. Finally, a ``received word''  $w$ is obtained by $w=\Enc(m) \oplus z$. This is defined formally below:

\begin{definition}[The experiment $\RNSY$]
We consider the following experiment (which we denote by $\RNSY$): The experiment works in two steps. The first step (denoted $\RNSY^1$) works as follows:
\begin{itemize}[noitemsep]
\item A message $m \from \B^k$ is chosen uniformly.
\item A restriction $\rho \from R^n_p$ is chosen uniformly.
\end{itemize}
The second step is defined for a fixed $m \in \B^k$ and $\rho \in R^n_p$. It is denoted by $\RNSY^2(m,\rho)$ and works as follows:
\begin{itemize}[noitemsep]
\item A string $z$ is chosen from $U^{1/3}_{p \cdot n}$.
\item We define $x=\Fill_{\rho}(z)$.
\item We define $w=\Enc(m) \oplus x$.
\end{itemize}
We use $(m,\rho,z,w) \from \RNSY$ to denote $m,\rho,z,w$ which are sampled by the two steps of this experiment.
\end{definition}

The next lemma fixes the coins of the local-decoder, making it deterministic. It is similar in spirit to Proposition 3.1 in \cite{RSV21}, except that the experiment $\RNSY$ that we use is different and more complicated than the one used in \cite{RSV21}.

\begin{lemma}
\label{lem:LLD to ARLLD}
There exists a constant $c>1$ such that if $n \ge \frac{c}{\eps^2}$ and $\Dec$ is a
$(\half-\eps,L,q,\delta)$-local list-decoder for a function $\Enc:\B^k \ar \B^n$ then there exists an oracle procedure $\overline{Dec}^{(\cdot)}(i,j)$ receiving $i \in [k]$, $j \in [L]$, and making $q$ queries to a string $w \in \B^n$, such that with probability at least $0.51$ over choosing $(m,\rho) \from \RNSY^1$, we have that with probability at least $0.01$ over choosing $(z,w) \from \RNSY^2(m,\rho)$,
there exists $j \in [L]$ such that
\[ \Pr_{i \from [k]}[\overline{Dec}^w(i,j)=m_i] \ge 1-10 \cdot \delta. \]
\end{lemma}

\begin{proof}
Let $\Dec$ denote an LLD for $\Enc$.
When choosing $\rho \from \Rest^n_p$ we expect that $\frac{1}{n} \cdot |\set{i:\rho(i)=1}|=\half-\frac{p}{2}$. Therefore, by a Chernoff bound:
\[ \Pr_{\rho \from \Rest^n_p}[ \frac{1}{n} \cdot|\set{i:\rho(i)=1}| \ge \half-0.49 \cdot p] \cdot  \le 2^{-\Omega(p^2 \cdot n)},\]
which is smaller than $0.01$, by our assumption that $n$ is sufficiently larger than $1/\eps^2$.
If this event occurs, then for every $z \in \B^{n \cdot p}$,
\[\wt(\Fill_{\rho}(z)) \le (\half-0.49 \cdot p) \cdot n + \wt(z). \]
This means that if $\wt(z) \le 0.34 \cdot p \cdot n$, then
\[ \wt(\Fill_{\rho}(z)) \le n \cdot (\half- (0.49-0.34) \cdot p) \le n \cdot(\half- \eps), \]
By our assumption that $p \le \frac{\eps}{10}$. Applying a Chernoff bound, we conclude that:
\[ \Pr_{z \from U^{1/3}_{n \cdot p}}[\wt(z) \ge 0.34 \cdot p \cdot n] \le 2^{-\Omega(pn)}, \]
which is smaller than $0.01$ by the assumption that $n$ is sufficiently larger than $1/\eps^2$.

Overall, this means that for $\gamma=0.02$, with probability at least $1-\gamma$ over choosing $(m,\rho,z,w) \from \RNSY$, we have that $\mathsf{dist}(\Enc(m),w) \le \half-\eps$, meaning that $m \in \List^{\Enc}_{\half-\eps}(w)$. By the definition of LLD, this gives that whenever this occurs, with probability at least $2/3$ over the choice of $r^{\shared}$, there exists $j \in [L]$ such that the procedure $P_{w,j,r^{\shared}}(i,r)=\Dec^w(i,j,r^{\shared},r)$ locally computes $m$ with error~$\delta$.

Let $E_1$ be the experiment in which $(m,\rho,z,w) \from \RNSY$ and $r^{\shared}$ be an independent uniform string. It follows that:

\[\Pr_{E_1}[\exists j \in [L]: \mbox{$P_{w,j,r^{shared}}$ locally computes $m$ with error $\delta$}] \ge \frac{2}{3} - \gamma. \]

By averaging, there exists a fixed string $\hat{r}^{\shared}$ such that:
\[\Pr_{\RNSY}[\exists j \in [L]: \mbox{$P_{w,j,\hat{r}^{shared}}$ locally computes $m$ with error $\delta$}] \ge \frac{2}{3} - \gamma. \]

Let $S$ denote the set of quadruples $(m,\rho,z,w)$ in the support of $\RNSY$ for which the event above occurs. For every such quadruple, we have that there exists a $j \in [L]$ for which $P_{w,j,\hat{r}^{shared}}$ locally computes $m$ with error $\delta$. Let $f$ be a mapping that given a quadruple $(m,\rho,z,w) \in S$, produces such a $j \in [L]$. This means that:

\[\Pr_{\RNSY}[\mbox{$P_{w,f(m,\rho,z,w),\hat{r}^{shared}}$ locally computes $m$ with error $\delta$}] \ge \frac{2}{3} - \gamma. \]

Let $\RNSY'$ be the experiment in which $(m,\rho,z,w) \from (\RNSY|(m,\rho,z,w) \in S)$. Namely, we choose $(m,\rho,z,w)$ from the experiment $\RNSY$, conditioned on the event that $(m,\rho,z,w) \in S$.

Let $E_2$ be the experiment in which we choose independently a random string $r$, $i \from [k]$ and $(m,\rho,z,w) \from \RNSY'$.
We obtain that:

\[\Pr_{E_2}[\Dec^w(i,f(m,\rho,z,w),\hat{r}^{shared},r)=m_i] \ge 1-\delta, \]
since $P_{w,f(m,\rho,z,w),\hat{r}^{shared}}$ computes correctly each coordinate $m_i$ with probability at least $1-\delta$ over the choice of $r$.

By averaging, there exists a fixed string $\hat{r}$ such that:


\[\Pr_{(m,\rho,z,w) \from \RNSY', i \from [k]}[\Dec^w(i,f(m,\rho,z,w),\hat{r}^{shared},\hat{r})=m_i] \ge 1-\delta. \]

By Markov's inequality:

\[\Pr_{(m,\rho,z,w) \from \RNSY'}\left[\Pr_{i \from [k]}[\Dec^w(i,f(m,\rho,z,w),\hat{r}^{shared},\hat{r}) \ne m_i] \ge 10\delta\right] \le \frac{1}{10}. \]

Let $\overline{\Dec}^w(i,j)=\Dec^w(i,j,\hat{r}^{shared},\hat{r})$. We obtain that:

\[\Pr_{(m,\rho,z,w) \from \RNSY'}[\Pr_{i \from [k]}[\overline{\Dec}^w(i,f(m,\rho,z,w)) = m_i] > 1-10\delta] >\frac{9}{10}. \]

Which gives that:
\[\Pr_{(m,\rho,z,w) \from \RNSY}\left[\Pr_{i \from [k]}[\overline{\Dec}^w(i,f(m,\rho,z,w)) = m_i] > 1-10\delta\right] > \left(\frac{2}{3} - \gamma\right) \cdot \frac{9}{10}  > 0.55. \]

Let $A$ denote the event:
\[ A= \set{\Pr_{i \from [k]}[\overline{\Dec}^w(i,f(m,\rho,z,w)) = m_i] > 1-10\delta},  \]
so that we have that:
\[\Pr_{(m,\rho,z,w) \from \RNSY}[A] > 0.55. \]
By the definition of the two steps of $\RNSY$, and an averaging argument this gives that:
\[\Pr_{(m,\rho) \from \RNSY^1}\left[\Pr_{(z,w) \from \RNSY^2(m,w)}[A] > 0.51] > 0.01\right]. \]
This concludes the proof.
\end{proof}

Applying Lemma \ref{lem:LLD to ARLLD} on the $(\half - \eps, q'=O(q \log k), L, 1/20k)$-LLD that we have previously obtained, gives the following corollary:

\begin{claim}
\label{clm:ARLLD}
There exists an oracle procedure $\overline{Dec}^{(\cdot)}(i,j)$ receiving $i \in [k]$, $j \in [L]$, and making $q'=O(q \cdot \log k)$ queries to a string $w \in \B^n$, such that with probability at least $0.51$ over choosing $(m,\rho) \from \RNSY^1$, we have that with probability at least $0.01$ over choosing $(z,w) \from \RNSY^2(m,\rho)$,
there exists $j \in [L]$ such that for every $i \from [k]$, $\overline{Dec}^w(i,j)=m_i$.
\end{claim}

\begin{proof}
This follows directly from Lemma \ref{lem:LLD to ARLLD}, noticing that a
statement of the form:
\[ \Pr_{i \from [k]}[\overline{Dec}^w(i,j)=m_i] \ge 1-\frac{1}{2k}, \]
implies that for every $i \from [k]$, $\overline{Dec}^w(i,j)=m_i$.
\end{proof}

We can use this (in a similar way to Lemma \ref{lem:Red => C}) to argue that:

\begin{claim}
\label{clm:starting:code}
Let $a=\log L+100$ and let $\ell=\log k$.
For every $\alpha \in \B^a$, there exists a $q$-CNF $C^{\alpha}$ over $n$ variables, such that for $C=\bigvee_{\alpha \in \B^a} C^{\alpha}$, and $p=\frac{\eps}{10}$:
\begin{itemize}[noitemsep]
\item $\Pr[C(U_n)=1] \le 2^{-(2^{\ell}-a)}$.
\item $\Pr_{\rho \from \Rest^n_p}[\Pr_{z \from U^{1/3}_{n \cdot p}}[C(\Fill_{\rho}(z))=1] \ge 0.01] \ge 0.01$.
\end{itemize}
\end{claim}

\begin{proof}
The proof is essentially identical to that of Lemma \ref{lem:Red => C}. More precisely, the message $m \in \B^k$ plays the role of $f \in \B^{2^{\ell}}$, the index $i \in [k]$ plays the role of $x \in \B^{\ell}$, and the index $j \in [L]$, plays the role of $\alpha \in \B^a$. The argument in the proof of Lemma \ref{lem:Red => C} shows that for every $m \in \B^k$, there exists a circuit $C_m$ of the required form such that:
\begin{itemize}[noitemsep]
\item $\Pr_{m \from \B^k,x \from \B^n}[C_m(x)=1] \le 2^{-(2^{\ell}-a)}$.
\item With probability at least $0.51$ over choosing $(m,\rho) \from \RNSY^1$, we have that with probability at least $0.01$ over choosing $(z,w) \from \RNSY^2(m,\rho)$, we have that $C_m(x)=1$, for $x=\Fill_{\rho}(z)$ (as in experiment $\RNSY$).
\end{itemize}
By applying Markov's inequality on each one of the two items, we can obtain that:
\begin{itemize}[noitemsep]
\item $\Pr_{m \from \B^k}\left[\Pr_{x \from \B^n}[C_m(x)=1] \le 2^{-(2^{\ell}-a-100)}\right] \ge 1-2^{-100}$.
\item With probability at least $0.1$ over choosing $m \from \B^k$, we have that with probability at least $0.01$ over choosing $\rho \from \Rest^n_p$, we have that with probability at least $0.01$ over choosing $(z,w) \from \RNSY^2(m,\rho)$, we have that $C_m(x)=1$, for $x=\Fill_{\rho}(z)$.
\end{itemize}
This allows to do a union bound, and obtain that there exists $m \in \B^k$ such that setting $C=C_m$, meets the conclusion of Claim \ref{clm:starting:code}.
\end{proof}

We are finally ready to prove Theorem \ref{thm:main:code:small eps}. Using Lemma \ref{lem:continuing} we conclude that assuming $a+100 \le \frac{2^{\ell}}{10}$ (which for sufficiently large $k$, follows by our assumption that $L \le 2^{k/20}$) we get that $q' \ge \frac{1000}{\eps}$ which gives that $q =\Omega(\frac{1}{\eps \cdot \log k})$, as required.

\section{Conclusion and Open Problems}
\label{sec:open problems}

The most interesting open problem is Open Problem \ref{open problem:extreme}. We hope that Theorem \ref{thm:fix} may help to point us to new constructions.

However, it is possible that the answer to Open Problem \ref{open problem:exist proof?} is negative, showing that black-box proofs cannot be used to solve Open Problem \ref{open problem:extreme}. Is this the case? If it is, can we show this?

Easier problems towards showing a negative answer to Open Problem \ref{open problem:exist proof?} are:
\begin{itemize}[noitemsep]
\item Is it true that in any black-box $\HFtoPRG$ proof, the number of queries $q \ge m$ or even $q \ge m^2$? A positive answer will show that the cost of the hybrid argument (in terms of the number queries used by the reduction) is unavoidable in black-box $\HFtoPRG$ proofs.
\item We don't know whether a super-constant number of queries is necessary for constant $\eps$. Can we show a super-constant lower bound on the number of queries of a reduction for a black-box $\HFtoPRG$ proof?
\item In fact, we don't even know to show a $q>1$ lower bound for black-box $\rhoHFtoepsPRG{\rho}{\eps}$ proof for constant $\eps$, and $\rho = \half+\eps$. Can we show this?
\item In all cases above, the question is open even for small values of $a$ (that do not apply in the extreme high-end).
\end{itemize}

Another approach to hardness vs. randomness was very recently suggested by Chen and Tell \cite{CT2}. They use an assumption which is less standard, and incomparable to \textsc{the extreme high-end hardness assumption} to construct ``target PRGs'' which are weaker than PRGs, but suffice for fast derandomization of randomized algorithms. It is interesting to investigate the power of this approach. For more details on this approach and exciting recent developments in this area, see the survey paper \cite{CTsurvey}.

Finally, another open problem is to further improve our lower bounds on the number of queries for reductions for black-box hardness amplification. More specifically, our improved lower bounds for hardness amplification apply in the extreme high-end, but are not tight. We only get $q \ge \max(\Omega(\ell),\Omega(\frac{1}{\eps}))$, whereas the known upper bounds give $q=O(\frac{\ell}{\eps^2})$. Can we prove a matching lower bound that applies in the extreme high-end? (namely, for $a=2^{(1-o(1)) \cdot \ell}$). Such lower bounds were given by \cite{GSV18} for the high-end, namely for $a \le 2^{\nu \cdot \ell}$ for a constant $\nu>0$.

\section*{Acknowledgements}

We are grateful to anonymous referees for very helpful comments and suggestions.

\bibliographystyle{alpha}
\bibliography{bbPRGlimitation}

\end{document}